\documentclass[11pt]{amsart}
\usepackage[english]{babel}
\usepackage{amssymb,amsfonts,amsmath}
\usepackage{amscd}
\usepackage{setspace}
\usepackage{longtable}
\usepackage[noadjust]{cite} 
\usepackage{enumerate}
\usepackage{float}
\usepackage{euscript}
\usepackage[unicode]{hyperref}
\usepackage{graphicx}
\usepackage{ragged2e}
\usepackage{mathtools}
\usepackage{xfrac}
\usepackage{lipsum}
\usepackage[labelfont=small,labelformat=simple]{subcaption}
\usepackage{tikz-cd}

\graphicspath{{figures/}}
\usepackage{graphics}
\theoremstyle{plain}
\newtheorem{theorem}{Theorem}[section]

\newtheorem{lemma}[theorem]{Lemma}

\theoremstyle{definition}
\newtheorem{definition}[theorem]{Definition}
\newtheorem{remark}[theorem]{Remark}
\newtheorem{example}[theorem]{Example}

\begin{document}
\title[Parallel transport and monodromy]{Parallel transport along Seifert manifolds and fractional monodromy}
\author{N. Martynchuk}
\email{N.Martynchuk@rug.nl}
\author{K. Efstathiou}
\email{K.Efstathiou@rug.nl}

\address{Johann Bernoulli Institute for Mathematics and Computer Science\\
  University of Groningen\\ P.O. Box 407\\ 9700 AK Groningen\\ The Netherlands}

\begin{abstract}
The notion of fractional monodromy was 
introduced by Nekhoroshev, Sadovski\'{i} and Zhilinski\'{i} as a generalization of standard (`integer') monodromy in the sense of Duistermaat from torus bundles to singular 
torus fibrations. In the present paper we prove a general result that allows to compute fractional monodromy in various integrable Hamiltonian systems.
In particular, we show
that the non-triviality of fractional monodromy in $2$ degrees of freedom systems with a Hamiltonian circle action
is related only to the fixed points of the circle action. Our approach is based on the study of a specific notion of 
parallel transport along Seifert manifolds. \thanks{\ The final publication is available at Springer via \href{https://doi.org/10.1007/s00220-017-2988-5}{https://doi.org/10.1007/s00220-017-2988-5}.}
\end{abstract}

\maketitle

\newtheorem{fixme}{FIXME}

\newcommand{\setS}{\mathbb{S}}
\newcommand{\setR}{\mathbb{R}}
\newcommand{\setC}{\mathbb{C}}
\newcommand{\setT}{\mathbb{T}}
\newcommand{\setN}{\mathbb{N}}
\newcommand{\setZ}{\mathbb{Z}}
\newcommand{\SL}{\mathrm{SL}}

\newcommand{\moma}[1]{\ensuremath\begin{pmatrix}1&#1\\0&1\end{pmatrix}}

\newcommand*{\diagmatrix}[1]{
        \begin{matrix}  
                #1        & \cdots & 0\\
                \vdots & \ddots & \vdots\\
                0        & \cdots & #1
        \end{matrix}
}

\section{Introduction} \label{intro}

A fundamental notion in classical mechanics is the notion of \textit{Liouville integrability}. A Hamiltonian system 
\begin{equation*}
\dot{x} = X_H, \ \ \omega(X_H, \cdot) = - dH,
\end{equation*}
on the $2n$-dimensional symplectic manifold $(M, \omega)$ is called Liouville integrable if there exist almost everywhere independent functions
$F_1 = H, \ldots, F_n$ such that all \textit{Poisson brackets} vanish:
$$
\{F_i,F_j\} = \omega(X_{F_i}, X_{F_j}) = 0.
$$
Various Hamiltonian systems, such as the Kepler and two-centers problem, the problem of $n \le 3$ point vortices, Euler, Lagrange and Kovalevskaya tops, 
are integrable in this sense.

The topological significance of Liouville integrability reveals itself in the
Arnol'd-Liouville theorem \cite{Arnold1968}. Assume that  the \textit{integral map}
$$
F = (F_1, \ldots, F_n) \colon M \to \mathbb R^n
$$
is proper.
The theorem states that a tubular neighborhood of a connected regular fiber $F^{-1}(\xi_0)$ 
is
a trivial torus bundle $D^n \times T^n$ admitting (semi-local) \textit{action-angle coordinates} 
$$I \in D^n \ \mbox{ and } \ \varphi \mod 2\pi \in T^n, \ \ \omega = dI \wedge d\varphi.$$ 
In particular, $F$ is a singular torus fibration and on each torus  $\{\xi\} \times T^n$ the motion is quasi-periodic. 

The question whether and when the action-angle coordinates exist globally  was answered in \cite{Nekhoroshev1972, Duistermaat1980}. It turns out \cite{Nekhoroshev1972} that  global action-angle
coordinates exist if
the set $R \subset image(F)$ of regular values of $F$ is such that
$$
\pi_1(R, \xi_0) = 0 \ \mbox{ and} \ H^2(R, \mathbb R) = 0.
$$
Obstructions that give necessary and sufficient conditions for the existence of  global action-angle
coordinates were given by Duistermaat; see \cite{Duistermaat1980, Lukina2008}. One such obstruction
is called (\textit{standard}) \textit{monodromy}. It 
appears only if $\pi_1(R,\xi_0) \ne 0$ and entails the non-existence of global action coordinates. 

Since the work of Duistermaat,
standard monodromy has been observed in many integrable systems of classical mechanics as well as in integrable 
approximations to molecular and atomic systems. In the typical case of $n = 2$ degrees of freedom non-trivial monodromy is manifested by the presence of the 
so-called \textit{focus-focus} points of the integral fibration $F$ \cite{Lerman1994, Matveev1996, Zung1997}. 
Such a result is often referred to as \textit{geometric monodromy theorem}. It has been recently
observed in \cite{Efstathiou2017} that the geometric monodromy theorem is a consequence of the following topological result.

\begin{theorem} \label{theorem/fth} \textup{(\cite[\S 4.3.2]{Bolsinov2004}, \cite{Efstathiou2017})}
Assume that $n = 2$ and that $F \colon F^{-1}(R) \to R$ is invariant under a free fiber-preserving $\mathbb S^1$ action. 
For  a simple closed curve $\gamma \subset R$ set $E = F^{-1}(\gamma)$ and $B = E / \mathbb S^1$.  Then
 the
monodromy of the $2$-torus bundle $F \colon E \to \gamma$ is given by
$$
\begin{pmatrix}
 1 & \langle {\bf e}, B \rangle \\
 0 & 1
\end{pmatrix} \in \mathrm{SL}(2,\mathbb Z),
$$
where ${\bf e}$ is the Euler class of the principal circle bundle $\rho \colon E \to B$.
\end{theorem} 

\begin{remark} \label{remark/focus}
The number $\langle {\bf e}, B \rangle$, which is obtained by integrating the Euler class over the base $B$, is called
the \textit{Euler} (or the \textit{Chern}) \textit{number} of the principal circle 
bundle $\rho \colon E \to B$.
Theorem~\ref{theorem/fth}
tells us that this Euler number 
determines the monodromy of the $2$-torus bundle $F \colon E \to \gamma$ and vice versa.

In  a neighborhood of the focus-focus fiber there exists a unique (up to orientation) system preserving $\mathbb S^1$ action that is 
free outside focus-focus points \cite{Zung1997}. On a small $3$-sphere $S^3_{\varepsilon}$ around a focus-focus point it defines a circle bundle with 
the Euler number $\langle {\bf e}, S^3_{\varepsilon} \rangle = 1$, the so-called \textit{anti-Hopf fibration}.
 It follows from Stokes' theorem  that 
for a small loop  $\gamma$
around the focus-focus critical value the Euler number $\langle {\bf e}, B \rangle$ 
equals the number of focus-focus points on the singular fiber. In particular,
 standard monodromy along $\gamma$ is non-trivial. More details can be found in \cite{Efstathiou2017}; see also 
 Section~\ref{sec/PTSM}. 
\end{remark}

Even though standard monodromy is related to singularities of the torus fibration $F$, it is an invariant of the regular, non-singular part $F \colon F^{-1}(R) \to 
R$.
An invariant that generalizes standard monodromy to singular torus fibrations is called \textit{fractional monodromy} \cite{Nekhoroshev2006}. 
We note that fractional monodromy is not a complete invariant of such fibrations --- it contains less information than the
\textit{marked molecule} in Fomenko-Zieschang theory \cite{Fomenko1990, Bolsinov2004, Bolsinov2012} --- but it is important for
applications and appears, for instance, in the so-called $m${:}$(-n)$ \textit{resonant systems}
\cite{Nekhoroshev2006, Nekhoroshev2007, Sugny2008, Schmidt2010, Efstathiou2013}; see Section 4.1 for details.

It was observed by Bolsinov \textit{et al.} \cite{Bolsinov2012}  that in 
$m${:}$(-n)$ resonant systems 
the circle action defines a 
Seifert fibration on a small $3$-sphere  around the equilibrium point and that the Euler number of this fibration equals the number appearing in the matrix of fractional 
monodromy, cf. Remark~\ref{remark/focus}.  
The question that remained unresolved  is why  this equality holds. In the present paper we give a complete answer to this question by proving the following results.

(i) Parallel transport and, therefore, fractional monodromy can be naturally
defined for closed Seifert manifolds (with an orientable base of genus $g > 0$). 

(ii) The fractional monodromy matrix is given by the Euler number of the
associated Seifert fibration. In the case of integrable systems, this Euler number
can be computed in terms of the fixed points of the circle action.

The latter result generalizes the corresponding results  of \cite{Efstathiou2013, Efstathiou2017} and, in particular, Theorem~\ref{theorem/fth}, thus
demonstrating that for standard and fractional monodromy the circle  action is more
important than the precise form of the integral map $F$.

We note that the importance of Seifert fibrations in integrable systems was discovered by Fomenko and Zieschang.
In their classification theorem \cite{Fomenko1990, Bolsinov2004}  Seifert manifolds play a central role: regular isoenergy surfaces of integrable nondegenerate  systems
 with $2$ degrees of freedom admit decomposition into families, each of which has a natural structure of a Seifert fibration. In our case of a global circle action
 there is only one such family, which has a certain label associated to it, the so-called $n$-\textit{mark} \cite{Bolsinov2004}. In fact, this $n$-mark coincides with the 
 Euler number that appears in Theorem~\ref{theorem/fth} and is related to the Euler number in the general case; see Remark~\ref{remark/mark}.  Our results therefore show how exactly this $n$-mark
 determines fractional monodromy.

\subsection{$1$:$(-2)$ resonant system} \label{sec/intrresonance}
Here, as a preparation to the more general setting of Sections~\ref{sec/PTSM} and \ref{sec/generalizedmonodromy}, we discuss the
famous example of a Hamiltonian system with fractional monodromy due to Nekhoroshev, Sadovski\'{i} and Zhilinski\'{i} \cite{Nekhoroshev2006}.

Consider $\setR^4$ with the standard symplectic structure $\omega = dq \wedge dp$. Define the \textit{energy} by
\begin{equation*}
H = 2q_1p_1q_2 + (q_1^2 - p_1^2)p_2 +  R^2, 
\end{equation*}
where $R = \frac{1}{2}(q_1^2 + p_1^2) + (q_2^2 + p_2^2)$, and the
 \textit{momentum} by
\begin{equation*}
J = \dfrac{1}{2}(q_1^2 + p_1^2) - (q_2^2 + p_2^2).
\end{equation*}
\begin{figure}[htbp]
  \includegraphics[width=\linewidth]{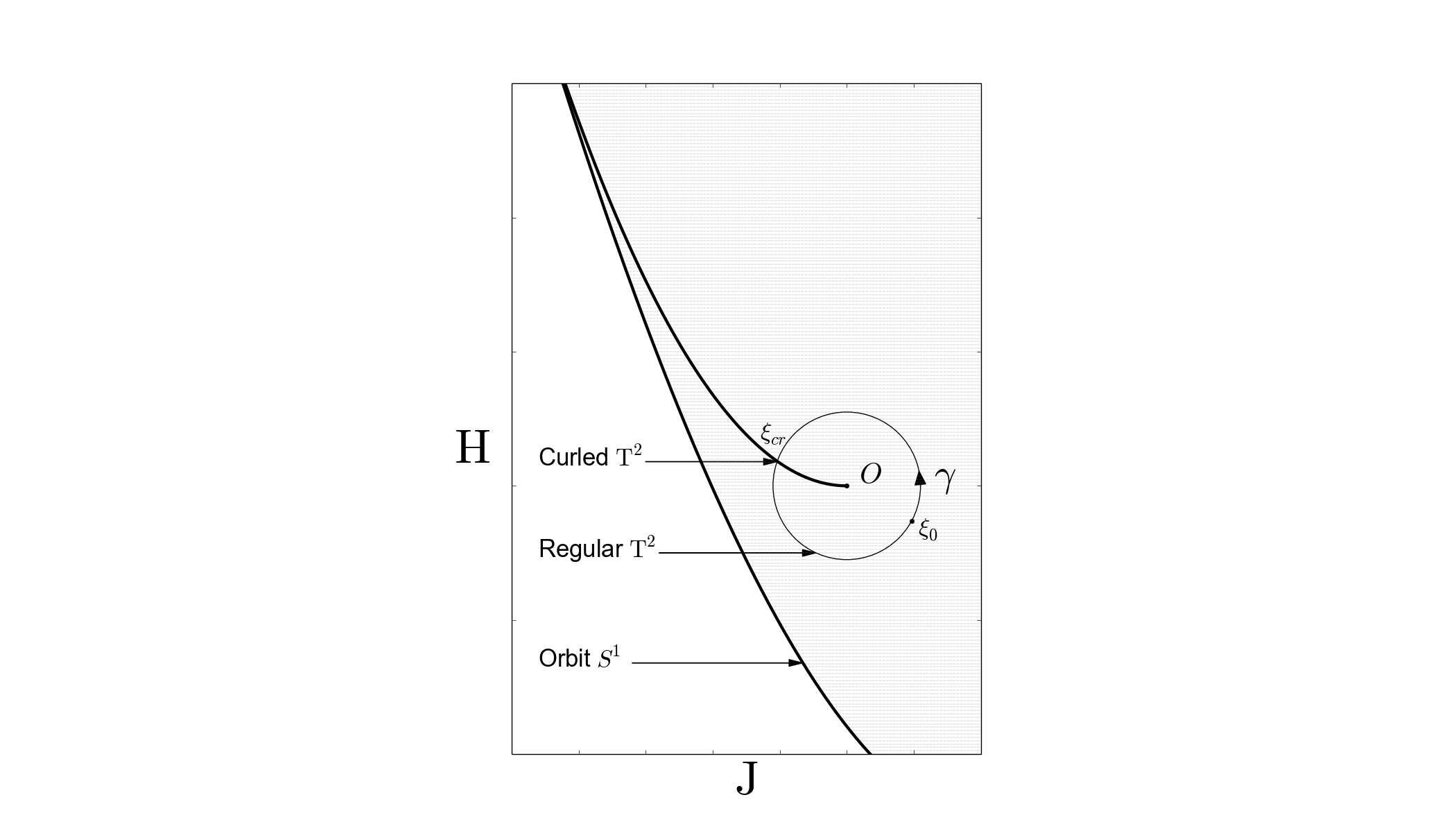}
  \caption{The bifurcation diagram of the  $1$:$(-2)$  resonant system.
   Critical values are colored black; the set $R$ is shown gray;
    the closed curve $\gamma$ around the origin intersects the hyperbolic branch of critical values once and transversely.}
  \label{1stBD}
\end{figure}

A straightforward computation shows that the functions $H$ and $J$ Poisson commute,
so the map $F = (J,H)$ defines an integrable Hamiltonian system on $\mathbb R^4$. 
The \textit{bifurcation diagram}, that is, the set of critical values of $F$,
is depicted in Figure~\ref{1stBD}. 
Let
$$R  = \{\xi \in \textup{image}(F) \mid \xi \mbox{ is a regular value of } F\}.$$
Since the fundamental group of the set $R$ vanishes, there is no monodromy and, thus, the $2$-torus bundle 
$F \colon F^{-1}(R) \to R$
admits a free fiber-preserving action of a $2$-torus. We observe that the  acting torus contains a subgroup $\mathbb S^1$ whose action extends from
$F^{-1}(R)$ to the whole $\mathbb R^4$. Indeed, such an action is given by the Hamiltonian flow of $J$. In complex coordinates
$z = p_1 + iq_1$ and $w = p_2+iq_2$ it has the form
\begin{equation} \label{circleaction}
(t,z,w) \mapsto (e^{it}z, e^{-2it}w), \ t \in \setS^1.
\end{equation}
From above it follows that the $\mathbb S^1$ action is
free on $F^{-1}(R)$ and, moreover, has a trivial
Euler class. However, on the whole phase space $\mathbb R^4$ the action is no longer trivial: the origin is fixed and the punctured plane
$$
P = \{(q,p) \mid q_1 = p_1 = 0 \mbox{ and } q_2^2 + p_2^2 \ne 0\}
$$
consists of points with $\mathbb Z_2$ isotropy group. This implies that the Euler number of the Seifert
$3$-manifold $F^{-1}(\gamma)$, where $\gamma$ is as in Fig.~\ref{1stBD}, equals $1/2 \ne 0$. Indeed, Stokes' theorem implies that
the Euler number of
$F^{-1}(\gamma)$ coincides with the Euler number of a small $3$-sphere around the origin $z = w = 0$. The latter Euler number equals $1/2$
because of  \eqref{circleaction}.

The following result
shows that the non-trivial Euler number of the Seifert manifold $F^{-1}(\gamma)$ enters the monodromy context, giving rise to what is now known as fractional monodromy.

\begin{lemma} \label{lemma/rth}
 Let $\mathbb Z_2 = \{1,-1\}$ denote the order two subgroup of the acting circle $\mathbb S^1.$ The quotient space
 $F^{-1}(\gamma) / \mathbb Z_2$ is the total space of a torus bundle over $\gamma$. Its standard monodromy is given by
 $
\begin{pmatrix}
 1 & 1 \\
 0 & 1
\end{pmatrix} \in \mathrm{SL}(2,\mathbb Z).
$

\end{lemma}
\begin{proof}
 Let $\xi \in \gamma \cap R$. Then the fiber $F^{-1}(\xi)$ is a $2$-torus. Since the $\mathbb S^1$ action is free on this fiber,
 the quotient $F^{-1}(\xi) / \mathbb Z_2$ is a $2$-torus as well. 
 
 Consider the critical value $\xi_{cr} \in \gamma.$
 Its preimage  $F^{-1}(\xi_{cr})$ is the so-called \textit{curled torus}; see Fig.~\ref{curledtorus}.
 \begin{figure}[htbp]
\newlength{\halfwidth}
\setlength{\halfwidth}{0.4\linewidth}

\hspace{0.38cm}
\begin{subfigure}[b]{\halfwidth}
  \includegraphics[width=\halfwidth]{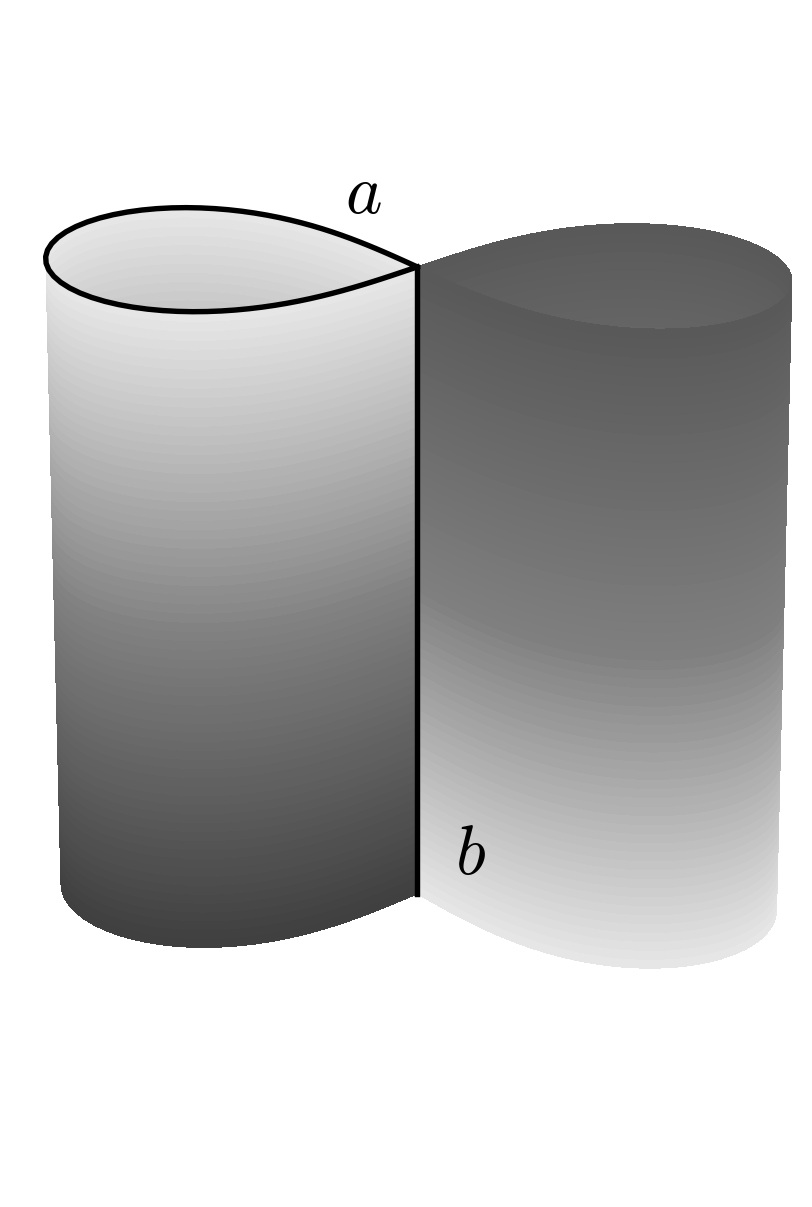}
  \caption{Cycles $(a,b)$}
  \label{homcycles}
\end{subfigure}
\hspace{0.9cm}
\begin{subfigure}[b]{\halfwidth}
  \includegraphics[width=\halfwidth]{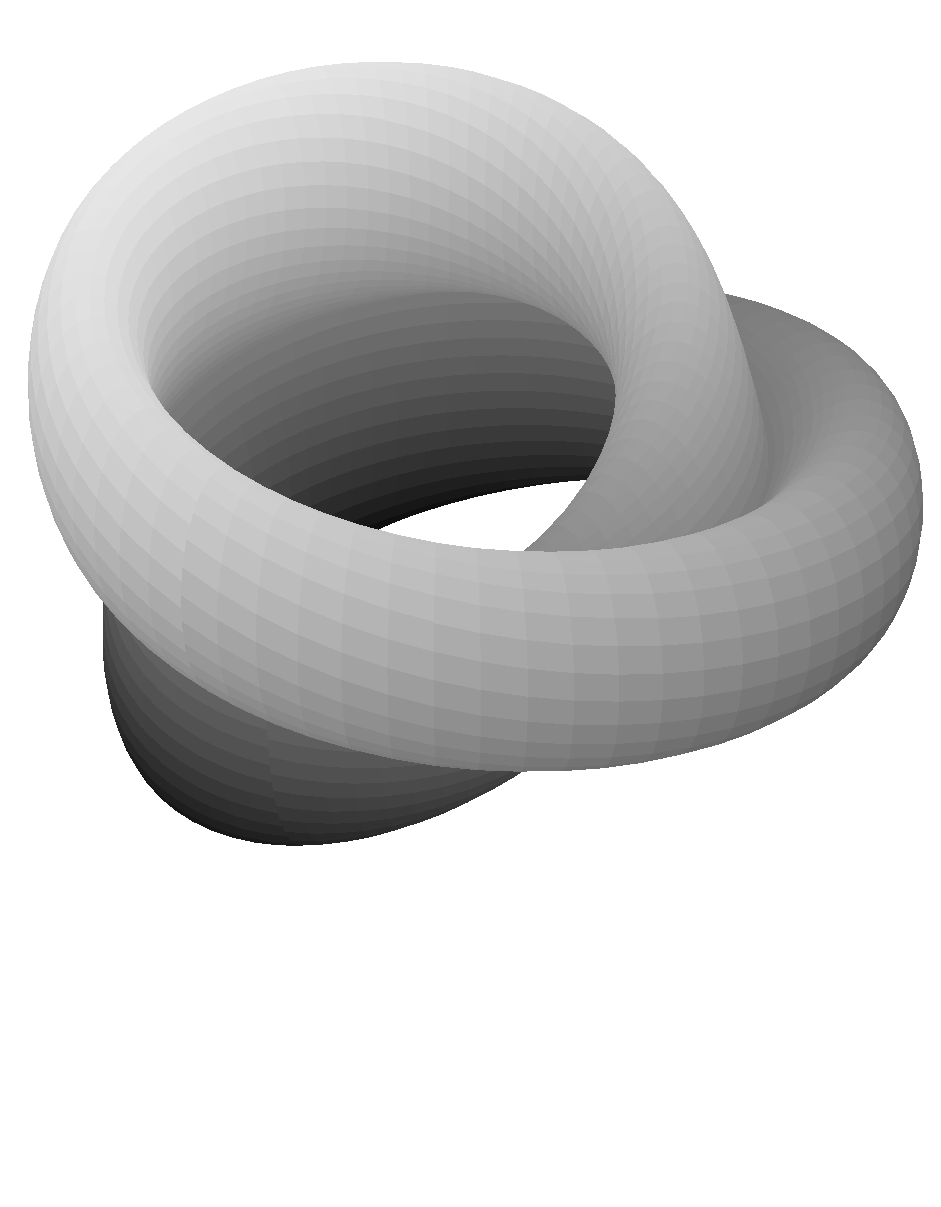}
  \caption{Curled torus}
  \label{curledtorus}
\end{subfigure}
\caption{Representation of a curled torus. Take a cylinder over the figure `eight', as shown in Figure $(a)$. 
Glue the upper and lower halves of this cylinder after rotating the upper part by $\pi$. The resulting surface is a curled torus $(b)$. }
\end{figure}
In this case there is a `short' 
orbit $b$ of the $\mathbb S^1$ action, formed by the fixed points of the  $\mathbb Z_2$ action.  The `short' orbit passes through the tip of the cycle $a$; see Fig.~\ref{homcycles}.
Other orbits are `long', that is, principal. From this description it follows that
after taking the $\mathbb Z_2$ quotient  only half of the cylinder survives and, thus,
$F^{-1}(\xi_{cr}) / \mathbb Z_2$ is topologically a $2$-torus. In view of \cite{Fomenko1990}, we have shown that
$$F \colon F^{-1}(\xi) / \mathbb Z_2 \to \gamma$$
is a torus bundle.
In order to complete the proof of the theorem it is left to apply Theorem~\ref{theorem/fth}. Indeed, since 
the Euler number of $F^{-1}(\gamma)$ equals $1/2$,
the Euler number of $F^{-1}(\gamma) / \mathbb Z_2$ equals $1$.
\end{proof}

\begin{remark} \label{remark/mark}
 Lemma~\ref{lemma/rth} can be reformulated by saying that the $n$-mark of the loop molecule  associated to $\gamma$ equals $1$. The molecule has the form shown in
 Fig.~\ref{Molecule}. Note that the $A^*$ atom corresponds to the curled torus Fig.~\ref{curledtorus}. A similar statement holds for higher-order resonances.
 \begin{figure}[htbp]
\hspace{0cm}
  \includegraphics[width=0.6\linewidth]{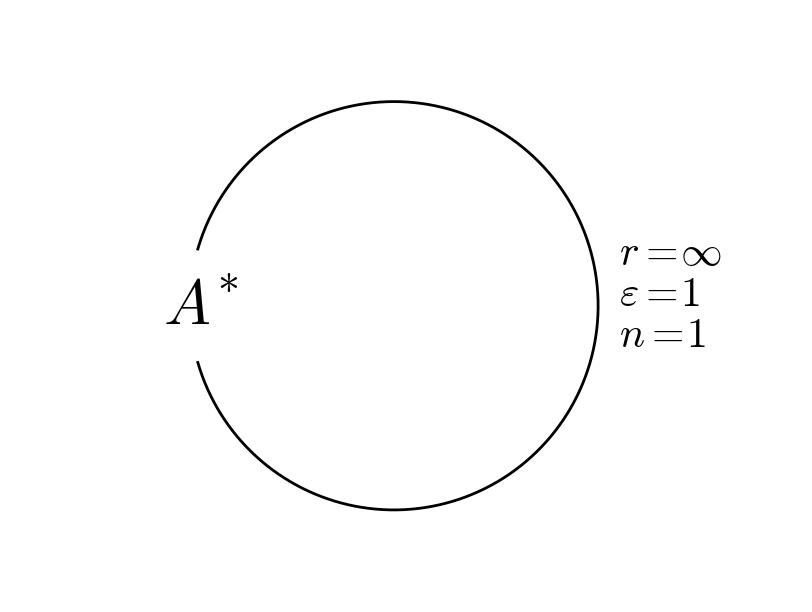}
  \caption{The loop molecule associated to $\gamma$.}
  \label{Molecule}
\end{figure}
\end{remark}

\begin{remark}
 We note that the symplectic structure on an open neighborhood $O$ of the manifold $F^{-1}(\gamma)$ does not descend to the $\mathbb Z_2$-quotient. Hence,
 the fibration
 $F \colon O /\mathbb Z_2 \to \mathbb R^2$ does not carry a natural Lagrangian structure
and Duistermaat's monodromy (parallel transport) along $\gamma$ is not defined. Instead, we use the more general Definition~\ref{dpt0}.
\end{remark}

Cutting the manifold $F^{-1}(\gamma) / \mathbb Z_2$ along any fiber $F^{-1}(\xi_0)/\mathbb Z_2, \, \xi_0 \in \gamma \cap R,$ we
get a manifold $X$ with the boundary $\partial X = X_0 \sqcup X_1$ consisting of the two tori $X_i.$
Following \cite{Efstathiou2013}, we define the parallel transport 
using the connecting homomorphism of the long exact sequence of the pair $(X, \partial X)$.

\begin{definition} \label{dpt0}
  The cycle $\alpha_1 \in H_1(X_1)$ is a \textit{parallel transport} of the cycle $\alpha_0 \in H_1(X_0)$ \textit{along} $X$ if  $$(\alpha_0, -\alpha_1) \in \partial_{*}(H_2(X, \partial X)),$$ 
  where $\partial_{*}$ is the connecting homomorphism of the exact sequence
   \begin{equation*}
    \cdots  \rightarrow H_{2}\left(X\right)\rightarrow H_{2}\left(X, \partial X\right)\xrightarrow{ \\ \partial_{*} \\ } 
    H_{1}\left(\partial X\right)\rightarrow H_{1}\left(X\right)\rightarrow \cdots 
  \end{equation*}
\end{definition}
\begin{figure}[htbp]
\hspace{0cm}
  \includegraphics[width=0.7\linewidth]{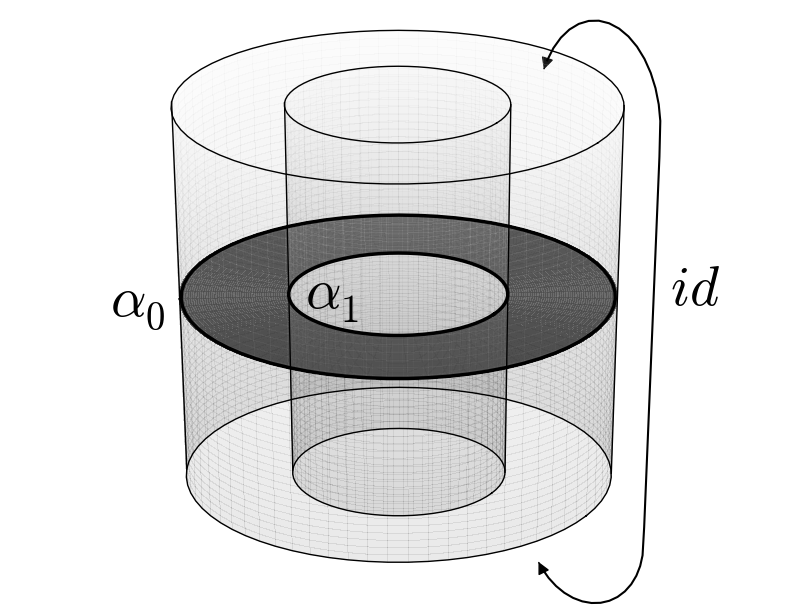}
  \caption{Parallel transport along $X$.}
  \label{PT}
\end{figure}

\begin{remark}
 Definition~\ref{dpt0} is applicable to an arbitrary manifold $X$ with boundary $\partial X = X_0 \sqcup X_1$. 
 For compact $3$ manifolds it may be reformulated as follows (see \cite{Hatcher2000}): $\alpha_1$ is a parallel transport of $\alpha_0$ along $X$ if 
 there exists an oriented $2$-dimensional submanifold $S \subset X$
 that `connects' $\alpha_0$ and $\alpha_1$: 
 $$\partial S = S_0 \sqcup S_1 \mbox{ and } [S_i] = (-1)^i\alpha_i \in H_1(X_i);$$
 see Fig.~\ref{PT}. We note, however, that even for compact $3$-manifolds  it might happen that, for a given homology cycle,
the parallel transport is not defined or is not  unique. 
For manifolds  $F^{-1}(\gamma)$ and $F^{-1}(\gamma)/\mathbb Z_2$ (and, more generally, for Seifert manifolds)  the parallel transport is unique; see Theorem~\ref{mth}. 
\end{remark}

From Lemma~\ref{lemma/rth} we infer that, in a homology basis of the fiber $F^{-1}(\xi_0)/\mathbb Z_2$, the 
parallel transport has the form of the monodromy matrix
\[
\begin{pmatrix}
 1 & 1 \\
 0 & 1
\end{pmatrix} \in \mathrm{SL}(2,\mathbb Z).
\]
For the fibration $F^{-1}(\gamma) \to \gamma$ this manifests the presence of nontrivial fractional monodromy.

\begin{theorem} \textup{(\cite{Nekhoroshev2006})} \label{theorem/rth}
 Let $(a_0,b_0)$ be   an integer basis of $H_1(F^{-1}(\xi_0))$, where $b$ is given by any orbit of the $\mathbb S^1$ action.
 The parallel transport is unique and has the form $2 a_0 \mapsto 2 a_0 + b_0$ and $b_0 \mapsto b_0$.
\end{theorem}

\begin{remark}
When written formally in an integer basis $(a_0,b_0)$, parallel transport has the form of a \textit{rational} matrix
\[
\begin{pmatrix}
 1 & 1/2 \\
 0 & 1
\end{pmatrix} \in \mathrm{SL}(2,\mathbb Q),
\]
called the matrix of \textit{fractional} monodromy.
\end{remark}

 Since the pioneering work \cite{Nekhoroshev2006}, various proofs of Theorem~\ref{theorem/rth} appeared; see
 \cite{Efstathiou2007, Sugny2008, Broer2010, Tonkonog2013} and \cite{Efstathiou2013}.
 Our proof, which is based on the singularities of
 the circle action,
 shows that 
 \begin{itemize}
\item the fixed point $\boldsymbol{0} \in \setR^4$ of the $\mathbb S^1$ action and
\item the short orbit $b$ with $\setZ_2$ isotropy
\end{itemize}
manifest the presence of fractional monodromy in this $1{:}(-2)$ resonant system. A similar kind of result holds in 
a
general setting of Seifert manifolds;
see Section~\ref{sec/PTSM}, and, in particular, in the setting of Hamiltonian systems with $m{:}(-n)$ 
resonance; see Subsection~\ref{resonance}.

 \subsection{The paper is organized as follows}

In Section~\ref{sec/PTSM} we consider a general setting of Seifert fibrations. We show that the parallel transport 
along the total space of such a fibration is given by its Euler number and the orders of the exceptional orbits; see Theorems~\ref{mth}. In the case when 
a Seifert fibration admits an equivariant filling, the
Euler number is 
 given by the fixed points of the
circle action inside the filling manifold; see Theorem~\ref{generalthm}. 

In Section~\ref{application}, after discussing the concepts of standard and (more general) \textit{fractional} monodromy in integrable Hamiltonian systems, 
we apply the results of Section~\ref{sec/PTSM} to fractional monodromy in the $2$ degrees of freedom case; see Theorems~\ref{corexistence} and \ref{simplef}. 
These theorems specify the
subgroup of homology cycles that admit parallel transport, and give a formula
for the computation of the fractional monodromy. These results, moreover, demonstrate that for standard and fractional monodromy the circle action is more important
than the precise form of the integral map.

Examples are investigated in Section~\ref{examples}.
The proof of Theorem~\ref{mth}  is given in Section~\ref{seifert}. We conclude with a discussion in Section~\ref{discussion}.

\section{Parallel transport along Seifert manifolds} \label{sec/PTSM}

\subsection{Seifert fibrations}

In the present subsection we recall the notions of a Seifert fibration and its Euler number. For a more detailed exposition we refer to \cite{Fomenko2010}.

\begin{definition} \label{defsm}
 Let $X$ be a compact orientable $3$-manifold (closed or with boundary) which is
invariant under an effective fixed point free $\mathbb S^1$ action. Assume that the $\mathbb S^1$ action is free on the boundary $\partial X$. Then
$$ \rho \colon X \to B = X / \mathbb S^1 $$
is called a \textit{Seifert fibration}. The manifold $X$ is called a \textit{Seifert manifold}.
\end{definition}

\begin{remark}
 From the slice theorem \cite[Theorem I.2.1]{Audin2004} (see also \cite{Bochner1945}) it follows that the quotient $B = X / \mathbb S^1$ is an orientable 
 topological
 $2$-manifold. Seifert fibrations are also defined in a more general setting when the base $B$ is non-orientable; see \cite{Fomenko2010}, \cite{Jankins1983}. However, 
 in this case there is no $\mathbb S^1$ action and the parallel transport is not unique; see Remark~\ref{remarknsm}. We will therefore consider
 the orientable case only.
 
\end{remark}

 Consider a Seifert fibration 
 $$\rho \colon X \to B = X / \mathbb S^1$$
 of a closed Seifert manifold $X$.  Let $N$ be
 the least common multiple of the orders of the exceptional orbits, that is, the orders of non-trivial isotropy groups. Since $X$ is compact, the number $N$ is well defined.  Denote by 
 $\mathbb Z_N$ the order $N$ subgroup of the acting circle $\mathbb S^1$.
 The subgroup $\mathbb Z_N$ acts on the Seifert manifold $X$. We thus have the reduction map $h \colon X \to X' = X/ \mathbb Z_N$ and the commutative diagram
 $$
\begin{tikzcd}
X \arrow{r}{h} \arrow[swap]{d}{\rho} & X' \arrow{dl}{\rho'} \\
B & 
\end{tikzcd}
$$
with $\rho'$ defined via $\rho = \rho' \circ h$.
 By the construction,
 $\rho' \colon X'  \to B$
 is a principal circle bundle over $B$. We denote its Euler number by $e(X')$.

\begin{definition}
 The \textit{Euler number} of the Seifert fibration 
 $\rho \colon X \to B = X / \mathbb S^1$ is defined by $e(X) = e(X')/N.$
\end{definition}

\begin{remark} \label{remarku}
 We note that a closed Seifert manifold $X$ can have non-isomorphic $\mathbb S^1$ actions with
 different Euler numbers. Indeed, let $m$ and $n$ be co-prime integers. Consider the $\mathbb S^1$ action 
$$ 
(t,z,w) \mapsto (e^{imt}z, e^{-int}w), \ t \in \setS^1,
$$
on the $3$-sphere $S^3 = \{(z,w) \mid |z|^2 + |w|^2 = 1\}$. Then the Euler number of the fibration 
$\rho \colon S^3 \to S^3 / \mathbb S^1$
equals $1/mn$. Despite this non-uniqueness, we sometimes refer to $e(X)$ as the Euler number of the Seifert manifold $X$. This should not be a cause of confusion since it will be always clear from the context what is the underlying $\mathbb S^1$ action.
\end{remark}

In the following Subsection~\ref{mresult} we show that the Euler number of a Seifert fibration is an obstruction to the existence of a trivial parallel transport; see Definition~\ref{dpt0}.

\subsection{Parallel transport} \label{mresult}

Consider a Seifert fibration
$\rho \colon X \to B = X / \mathbb S^1$ such that the boundary  $\partial X = X_0 \sqcup X_1$ consists of two $2$-tori $X_0$ and $X_1$. 
Take an orientation and fiber preserving homeomorphism $f \colon X_0 \to X_1$. Any homology basis
 $(a_0,b_0)$ of $H_1(X_0)$ can be then mapped to the homology basis 
 $$(a_1 = f_{\star}(a_0), b_1 = f_{\star}(b_0))$$ 
 of $H_1(X_1).$ In what follows we assume that $b_0$ is equal to the homology class of a (any) fiber of the Seifert fibration on $X_0.$
Let 
$$X(f) = X / \sim, \ \ X_0 \ni x_0 \sim f(x_0) \in X_1, $$ 
be the closed Seifert manifold that is obtained from $X$ by gluing the boundary components using $f$.

Finally, let $N$ be the least common multiple of
$n_j$ -- the orders of the exceptional orbits.
With this notation we have the following result.

\begin{theorem} \label{mth}
  The parallel transport along $X$ is unique. Only linear combinations of $Na_0$ and $b_0$ can be parallel transported along $X$ and under the parallel
  transport
  \begin{equation*}
 \begin{split}
N a_0 & \mapsto N a_1 + k b_1 \\
b_0 & \mapsto b_1
\end{split}
\end{equation*} 
for some integer $k = k(f)$ which depends only on the isotopy class of $f.$ Moreover, the Euler number of $X(f)$ is given by 
$e(f) = k(f)/N.$
\end{theorem}
\begin{proof}
 See Section~\ref{seifert}.
\end{proof}

\begin{remark}
 We note that (by the construction) $X(f)/\mathbb S^1$ has genus $g > 0$ and hence is not a sphere. It follows that the $\mathbb S^1$ action on $X$ and $X(f)$
 is unique up to isomorphism; see
 \cite[Theorem 2.3]{Hatcher2000}.
\end{remark}

\begin{remark} \label{remarknsm}
  Even if the base $B$ is non-orientable, the group $\partial_*(X,\partial X)$ is still isomorphic to $\mathbb Z^2$.
  However, in this case, $\partial_*(X,\partial X)$ is spanned
  by $(b_0,b_1)$ and $(2b_0,0)$. It follows that 
  no multiple of $a_0$ can be parallel transported along $X$ and that the parallel transport is not unique. 
\end{remark}

\subsection{The case of equivariant filling}

Theorem~\ref{mth} shows that the Euler number of a Seifert manifold can be computed in terms of the parallel transport along this manifold.
But conversely, if we know the Euler number and the orders of exceptional orbits of a Seifert manifold, we also know how the parallel transport acts on homology cycles.
In applications the orders of exceptional orbits are often known.
In order to compute the Euler number one may then use the following result.

\begin{theorem} \label{generalthm}
Let $M$ be a compact oriented $4$-manifold that admits an effective circle action. Assume that the action is fixed-point free on the boundary
$\partial M$ and has only finitely many fixed points $p_1, \ldots, p_{\ell}$ in the interior. Then
\begin{equation*}
e(\partial M) = \sum\limits_{k=1}^{\ell} \dfrac{1}{m_k n_k},
\end{equation*} 
where $(m_k,n_k)$ are isotropy weights of the fixed points $p_k$.
\end{theorem}

\begin{remark}
Recall that near each fixed point $p_k$ the $\setS^1$ action can be linearized as
\begin{equation} \label{laction}
(t, z,w) \mapsto (e^{im_kt}z, e^{-in_kt}w), \ t \in \setS^1,
\end{equation}  
in appropriate coordinates $(z,w)$ that are positive with respect 
to the orientation of $M$. The isotropy weights $m_k$ and $n_k$ are co-prime integers. In particular, none of them
is equal to zero.
\end{remark}

\begin{remark}
In the above theorem neither $M$ nor $\partial M$ are assumed to be connected. The orientation on $\partial M$ is induced
by $M$. 
\end{remark}

\renewenvironment{proof}{\textit{Proof of Theorem~\ref{generalthm}.}}{\qed}

\begin{proof} 
Eq.~\eqref{laction} implies that for each fixed point $p_k$ there exists a small closed $4$-ball $B_k \ni p_k$ invariant under the action. 
Denote by $Z$ the manifold
$Z = M \setminus \bigcup_{k=1}^{\ell} B_k.$
Let 
$N$ be a common multiple of the orders of all exceptional orbits in $M$ and $\mathbb Z_N$ be the order $N$ subgroup of the acting
circle $\mathbb S^1$. Set
\begin{equation*}
X = Z / \setZ_N \ \ \ \mbox{ and }  \ \ \ Y = Z / \setS^1.
\end{equation*}
Denote by $\Pr \colon X \to Y$ the natural projection that identifies the orbits of the $\setS^1/\setZ_N$ action.
By the construction the triple $(X, Y, \Pr)$ is a principal circle bundle.

Because of the slice theorem \cite{Audin2004} the spaces $X$ and $Y$ are topological manifolds (with boundaries).  The boundary $\partial Y$ is a disjoint union of
the closed $2$-manifold $B = \partial M / \setS^1$ and the $2$-spheres $S^2_k = \partial B_k / \setS^1.$ 
Let $i_B \colon B \to Y$ and $i_k \colon S^2_k \to Y$ be the corresponding inclusions.

Denote by ${\bf e}_Y \in H^2(Y)$ the Euler class of the circle bundle $(X, Y, \Pr)$. 
By the functoriality $i^*_B({\bf e}_Y)$ and $i^*_k({\bf e}_Y)$ are the Euler classes of the circle bundles
$(\Pr^{-1}(B), B, \Pr)$ and $(\Pr^{-1}(S^2_k), S^2_k, \Pr)$, respectively. Hence
\begin{equation*}
\langle {\bf e}_Y, i_B(B) \rangle  = \langle i^{*}_B({\bf e}_Y),B \rangle = N e(\partial M)
\end{equation*}
 and analogously 
\begin{equation*}
\langle {\bf e}_Y, i_k(S^2_k) \rangle = \langle i^*_k{\bf e}_Y, S^2_k \rangle  = \dfrac{N}{m_k n_k}.
\end{equation*} 
The equality
\begin{equation*}
\langle {\bf e}_Y, i_B(B) - \sum_{k=1}^{\ell} i_k(S^2_k) \rangle = \langle {\bf e}_Y, \partial Y \rangle = 0
\end{equation*}   
completes the proof.
\end{proof}
\renewenvironment{proof}{\textit{Proof.}}{\qed}

\section{Monodromy in integrable systems} \label{sec/generalizedmonodromy}

\subsection{Historical and mathematical background} \label{mb}

\textit{Standard monodromy} was introduced by Duistermaat in \cite{Duistermaat1980} as an obstruction to the existence of global action coordinates in
integrable Hamiltonian systems. 
Since the early work \cite{Duistermaat1980},  non-trivial monodromy has been observed
in
the (quadratic) spherical pendulum
(\cite{Bates1993, Efstathiou2005}) \cite{Duistermaat1980, Cushman2015}, the Lagrange top \cite{Cushman1985}, the
Hamiltonian Hopf bifurcation \cite{Duistermaat1998}, the champagne bottle
\cite{Bates1991}, the
coupled angular momenta \cite{Sadovskii1999}, 
the hydrogen atom in crossed fields \cite{Cushman2000},
the two-centers problem \cite{Waalkens2003, Waalkens2004} and many other systems.
A common aspect of most of these systems is the presence of \emph{focus-focus} singular points of the  Lagrangian fibration. It is known that the presence of such singular points is 
sufficient for the monodromy to be nontrivial in the general case (\emph{geometric monodromy theorem}) \cite{Lerman1994, Matveev1996, Zung1997}.

The  definition of standard monodromy in the sense of Duistermaat \cite{Duistermaat1980} reads as follows.
Consider a Lagrangian $n$-torus bundle $F \colon M \to R$ over a $n$-dimensional manifold $R$. By definition, this means
that $M$ is a symplectic manifold and that each fiber $F^{-1}(\xi)$ is a Lagrangian submanifold of $M$.
\begin{remark}In the context of integrable systems
$R \subset \mathbb R^n$ and $F$ is given by $n$ Poisson commuting functions. Conversely, every chart $(V, \chi)$ of $R$ gives rise to an integrable system
on $F^{-1}(V) \subset M$ with the integral map $F \circ \chi$.
\end{remark}

There is a well-defined action of the fibers of $\Pr \colon T^{*}R \to R$ on the fibers of $F \colon M \to R$, which, in every chart $(V, \chi)$, is given by the flow of
  $n$ Poisson commuting functions $F \circ \chi$; see \cite{Duistermaat1980} and \cite{Lukina2008}. For each $\xi \in R$ the stabilizer of the
  $\mathbb R^n_\xi = \Pr^{-1}(\xi)$ action on $T^n_\xi = F^{-1}(\xi)$ is a lattice $\mathbb Z^n_\xi \subset \mathbb R^n_\xi$. 
  The union of these lattices
  covers the base manifold $R$:
  $$
  \Pr \colon \bigcup \mathbb Z^n_\xi \to R.
  $$
  \begin{definition}
  The (standard) monodromy of the Lagrangian  $n$-torus bundle $F \colon M \to R$ is defined as the representation
  \begin{align*}
 \rho \colon \pi_1(R,\xi_0) \to \textup{Aut}\,\mathbb Z^n_{\xi_0} \simeq \mathrm{GL}(n, \mathbb Z)
\end{align*}
of the fundamental group $\pi_1(R,\xi_0)$ of the base $R$ in the  group of automorphisms of $\mathbb Z^n_{\xi_0}$. For each element $[\gamma] \in \pi_1(R,\xi_0)$,
the automorphism $\rho([\gamma])$ is called the (standard) monodromy along $\gamma$
\end{definition}

\begin{remark} \label{Duistermaatpt}
 We note that the lattices $\mathbb Z^n_\xi$ give a unique local identification of cotangent spaces of $T^{*}R$, that is, a flat connection.
  Thus, standard monodromy is given by the parallel transport (holonomy) of this connection.
\end{remark}

The following lemma shows that, in the case of Lagrangian torus bundles, the parallel transport 
in the sense of Definition~\ref{dpt0} coincides with the  parallel transport of the flat connection, given in
Remark~\ref{Duistermaatpt}. 

\begin{lemma}
 Let $\gamma = \gamma(t)$ be a continuous curve and
 \begin{equation} \label{manx}
 X = \{(x,t) \in M \times [0,1] \colon F(x) = \gamma(t)\}.
\end{equation}
Then $(\alpha_0, -\alpha_1) \in \partial_{*}(H_2(X, \partial X))$
if and only if the
cycle $\alpha_1$ is a parallel transport of $\alpha_0$ in the sense of Remark~\ref{Duistermaatpt}.
\end{lemma}

\begin{proof}
 By homotopy invariance, we can assume that $\gamma$  is smooth.
 Let $(0 = t_0 \le \ldots \le t_n = 1)$ be a sufficiently fine partition
 of the segment $[0,1]$. Then, for each $i$, we have
 $$\gamma([t_i, t_{i+1}]) \subset V_i,$$ 
 where $V_i$ is a small open neighborhood $V_i \subset R$. By the Arnol'd-Liouville theorem \cite{Arnold1968}, the two notions of parallel transport
 along $\gamma|_{[t_i,t_{i+1}]}$ coincide. The result follows.
\end{proof}

\begin{remark}
 Let $\gamma$ be a simple curve. If $\gamma(0) \ne \gamma(1)$, then the manifold $X$ in \eqref{manx} is homeomorphic to $F^{-1}(\gamma)$. If 
 $\gamma(0) = \gamma(1) = \xi_0$, then the manifold $X$ is
 obtained from $F^{-1}(\gamma)$ by cutting along the fiber
 $F^{-1}(\xi_0)$.
\end{remark}

\textit{Fractional monodromy} was introduced in \cite{Nekhoroshev2006} as a generalization of standard monodromy in the sense of Duistermaat 
from Lagrangian torus bundles to singular Lagrangian fibrations. 
Since the pioneering work \cite{Nekhoroshev2006}, non-trivial fractional monodromy has been demonstrated in several integrable Hamiltonian
systems \cite{Nekhoroshev2007, Giacobbe2008, Sugny2008, Efstathiou2013}. 

What has been missing until now for fractional monodromy
 is a result that associates fractional monodromy to certain 
singular points of the Lagrangian fibration in the same spirit as the geometric monodromy theorem associates standard monodromy to focus-focus singular points.
In the next subsection~\ref{application} we give such a result for fractional monodromy in the case when
the fibration is invariant under an effective circle action. Specifically, we show that  fractional
monodromy is completely determined by the singularities of the corresponding circle action and that, in certain cases,
fractional monodromy can be computed in terms of the fixed points of this action,
just as  standard monodromy
\cite{Efstathiou2017}.

The  definition of fractional monodromy (in the sense of \cite{Nekhoroshev2006} and \cite{Efstathiou2013}) reads as follows.
Consider a singular Lagrangian fibration $F \colon M \to R$ over a $n$-dimensional manifold $R$, given by a proper integral map $F$.  
Locally, such a fibration gives an integrable Hamiltonian system.
Let $\gamma = \gamma(t)$ be a continuous closed curve in $F(M)$ such that the space
\begin{equation*}
 X = \{(x,t) \in M \times [0,1] \colon F(x) = \gamma(t)\}
\end{equation*}
is connected and such that  $\partial X = X_0 \sqcup X_1$
is a disjoint union of two regular tori $X_0 = F^{-1}(\gamma(0))$ and $X_1 = F^{-1}(\gamma(1)).$ Set
$$H^0_1 = \{\alpha_0 \in H_1(X_0) \mid \alpha_0 \mbox{ can be parallel transported along } X\}.$$

\begin{definition} \label{hfm}
If the parallel transport along $X$ defines an automorphism of the group $H^0_1$, then this automorphism is called
\textit{fractional monodromy along $\gamma$.}
\end{definition}

\begin{remark}
As was mentioned in Subsection~\ref{sec/intrresonance}, in the singular case the notion of parallel transport in the sense of Remark~\ref{Duistermaatpt} is not defined.
Instead, the more general Definition~\ref{dpt0} is used.
\end{remark}

\subsection{Applications to integrable systems} \label{application}
Consider a singular Lagrangian fibration $F \colon M \to R$ over a $2$-dimensional manifold $R$. Assume that the map  $F$ is proper
and invariant under an effective $\mathbb S^1$ action.  
Take a simple closed curve $\gamma = \gamma(t)$ in $F(M)$ that satisfies
the following regularity conditions:
\begin{enumerate}[(i)]
\item the fiber $F^{-1}(\gamma(0))$ is regular and connected;
\item  the $\setS^1$ action is fixed-point free on the preimage  $E = F^{-1}(\gamma)$;
\item the preimage $E$ is a closed oriented connected submanifold of $M$.
\end{enumerate}

\begin{remark}
 Note that, generally speaking, $F^{-1}(\gamma(t)), \ t \in [0,1],$ is neither smooth nor connected. 
\end{remark}

From the regularity conditions it follows that
\begin{equation*}
 X = \{(x,t) \in M \times [0,1] \colon F(x) = \gamma(t)\}
\end{equation*}
is a Seifert manifold with an orientable base. This manifold can be obtained from the Seifert manifold 
$E = F^{-1}(\gamma)$ by cutting along the fiber $F^{-1}(\gamma(0))$. We note that the boundary 
$\partial X = X_0 \sqcup X_1$ 
is a disjoint union of two tori.

Let $e(E)$ be the Euler number of $E$ and $N$ denote the least common multiple of $n_j$ -- the orders of the exceptional orbits.
  Take a basis $(a,b)$ of the homology group $H_1(X_0) \simeq \setZ^2$, where $b$ is given by any orbit of the $\setS^1$ action.
  Then the following theorem holds.
  
  \begin{theorem} \label{corexistence}
Fractional monodromy along $\gamma$ is defined. Moreover, $(Na,b)$ form a basis
 of the parallel transport group $H^0_1$  
 and the corresponding isomorphism has the form $b \mapsto b$ and $Na \mapsto Na + kb$, where $k \in \setZ$ is given by $k = Ne(E).$ 
\end{theorem}
\begin{proof}
 Follows directly from Theorem~\ref{mth}.
\end{proof}

\begin{remark}Theorem~\ref{corexistence} tells us that
 the orders of the exceptional orbits $n_j$ and the Euler number $e(E)$
 completely determine fractional monodromy along $\gamma$. 
 \end{remark}
 
 \begin{remark} 
Let $i_0 \colon X_0 \to X$ and $i_1 \colon X_1 \to X$ denote the corresponding inclusions. Observe that, in our case, the composition
\begin{equation*}
 i^{-1}_1 \circ i_0 \colon H_1(X_0, \mathbb Q) \to H_1(X_0, \mathbb Q)
\end{equation*}
gives an automorphism of the first homology group $H_1(X_0, \mathbb Q)$. In a basis of $H_1(X_0,\setZ)$ the
isomorphism $i^{-1}_1 \circ i_0$ is written as $2 \times 2$ matrix with rational coefficients, called the
 \textit{matrix of fractional monodromy} \cite{Tonkonog2013}. 
 We have thus proved that
  in a basis $(a,b)$ of $H_1(M_0)$, where $b$ corresponds to the $\setS^1$ action, the fractional monodromy matrix has the form
  $$ \begin{pmatrix}
            1 & e(E) = k/N \\
            0 & 1
           \end{pmatrix} \in \textup{SL}(2,\mathbb Q).$$
\end{remark}

In certain cases we can easily compute the parameter $e(E) = k/N,$ as is explained in the following theorem.
 
\begin{theorem} \label{simplef}
Assume that $\gamma$ bounds a compact $2$-manifold $U \subset R$ such that $F^{-1}(U)$ has only finitely many fixed points $p_1, \ldots, p_l$
of the $\setS^1$ action. Then
\begin{equation*}
e(E) = \sum\limits_{k=1}^l \dfrac{1}{m_k n_k},
\end{equation*} 
where $(m_k,n_k)$ are the isotropy weights of the fixed points $p_k.$
\end{theorem}
\begin{proof}
Follows directly from Theorem~\ref{generalthm}.
\end{proof}

\begin{remark}
For the case of standard monodromy, Theorem~\ref{simplef} agrees with Theorem 2.2 from \cite{Efstathiou2017}, which considers
only the case $m_k = 1$ and $n_k = \pm 1$ and which
states that the monodromy parameter is given by the sum of positive singular points ($n_k = 1$) of the
Hamiltonian $\setS^1$ action minus the number of negative singular points ($n_k = -1$).
\end{remark}

\begin{remark}
Theorem~\ref{generalthm}, when applied to the context of Lagrangian fibrations, tells us more than Theorem~\ref{simplef}.
Indeed,  consider smooth curves 
$\gamma_1$ and $\gamma_2$ that are cobordant in $R$. Theorem~\ref{generalthm} allows to compute
$$e(F^{-1}(\gamma_1)) - e(F^{-1}(\gamma_2)),$$
which is the difference 
between the Euler numbers of $F^{-1}(\gamma_1)$ and $F^{-1}(\gamma_2)$. This difference shows how far is
fractional monodromy along $\gamma_1$ from fractional monodromy along $\gamma_2$. Theorem~\ref{simplef} is recovered  when 
$\gamma_1$ is cobordant to zero.
\end{remark}

Combining Theorems~\ref{corexistence} and \ref{simplef} together one can compute fractional monodromy in various integrable Hamiltonian systems.
We illustrate this in the following Section~\ref{examples}.

\section{Examples} \label{examples}
 
\subsection{Resonant systems} \label{resonance}

In this section we consider $m$:$(-n)$ \textit{resonant systems} \cite{Nekhoroshev2007, Sugny2008, Schmidt2010, Efstathiou2013}, which are local 
models for integrable $2$ degrees of freedom systems with an effective Hamiltonian $\setS^1$ action. Our approach to these systems is very general. Moreover, it clarifies a question 
posed in \cite[Problem 61]{Bolsinov2012}, cf. Remark~\ref{remark/mark}.

\begin{definition}
Consider $\mathbb R^4$ with the canonical symplectic structure $dq \wedge dp$.
 An integrable Hamiltonian system 
 $$(\setR^4, dq \wedge dp, F = (J,H))$$
 is called a $m$:$(-n)$ \textit{resonant system} if the function
 $J$ is the $m$:$(-n)$ \textit{oscillator}
\begin{equation*}
  J = \dfrac{m}{2}(q_1^2 + p_1^2) - \dfrac{n}{2}(q_2^2 + p_2^2).
\end{equation*}
Here $m$ and $n$ be relatively prime integers with $m > 0$.
\end{definition}
We note that for every $m$:$(-n)$ resonant system there exists an associated effective $\setS^1$ action that preserves the integral map $F = (J,H)$.
Indeed, the induced Hamiltonian flow of $J$ is periodic. In coordinates $z = p_1 + iq_1$ and $w = p_2+ iq_2$ the action has the form
\begin{equation} \label{rescaction}
(t,z,w) \mapsto (e^{imt}z, e^{-int}w), \ t \in \setS^1.
\end{equation}

Assume that the integral map $F = (J,H)$ is proper. 
Let
$\gamma = (J(t),H(t))$ be a simple closed curve satisfying
the assumptions (i)-(iii) from Section~\ref{application}.
\begin{remark}We note that, in this case, the assumptions (i)-(iii) can be reduced to the following more easily verifiable conditions
\begin{enumerate}[(i')]
\item the fiber $F^{-1}(\gamma(0))$ is regular and connected;
\item the preimage $E = F^{-1}(\gamma)$ is connected;
\item 
for all $t$ the following holds:
$H'(t)dJ - J'(t)dH \ne 0. $
\end{enumerate}
\begin{proof}
Under (i)-(iii), the space $E = F^{-1}(\gamma)$
is the boundary
of the compact oriented manifold $F^{-1}(U)$, where $U$ is the $2$-disk bounded by $\gamma$. Hence, $E$ is itself compact and oriented. It is left to note that
the $\mathbb S^1$ action is fixed-point free on $E$.
\end{proof}
\end{remark}

Let $(a,b)$ be a basis of the integer homology group $H_1(F^{-1}(\gamma(0))$ such that $b$ is given by any orbit of the $\setS^1$ action.
There is the following result (cf. \cite{Efstathiou2013}).
 
\begin{theorem} \label{frmresonance} 
Let $U$ be a $2$-disk in the $(J,H)$-plane such that $\partial U = \gamma$. 
Case 1: $(0,0) \in U$. The parallel transport group is spanned by $mna$ and $b$.
The matrix of fractional monodromy has the form
 $$ \begin{pmatrix}
            1 & 1/mn \\
            0 & 1
           \end{pmatrix} \in \textup{SL}(2, \mathbb Q).$$
Case 2: $(0,0) \notin U$. The parallel transport group $H^0_1$ is spanned by $Na$ and $b,$ \ \ \  where $N \in \{1,m,n,mn\}$. The matrix of fractional monodromy is trivial. 
\end{theorem}
\begin{proof}
 In view of Theorems~\ref{corexistence} and \ref{simplef}, we only need to determine the least common multiple $N.$
 
\textit{Case 1.} In this case the fixed point $q = p = 0$ of the
 $\mathbb S^1$ action
  belongs to $F^{-1}(U) \subset \mathbb R^4.$
 Orbits with $\mathbb Z_m$ and $\mathbb Z_n$ isotropy group emanate from this fixed point and necessarily `hit' the boundary
 $F^{-1}(\gamma)$. It follows that the least common multiple is $N = mn$.
 
 \textit{Case 2.} In this case the fixed point $q = p = 0$ of the
 $\mathbb S^1$ action does not
  belong to $F^{-1}(U) \subset \mathbb R^4.$ However, $\gamma$ might intersect critical values of $F$ that give rise to exceptional orbits in
 $E = F^{-1}(\gamma)$
 with $\mathbb Z_{m}$ or $\mathbb Z_n$ isotropy group.  It follows that the least common multiple is $N = 1$, $m$, $n$ or $mn$.
\end{proof}
\begin{figure}[ht]
\includegraphics[width=\linewidth]{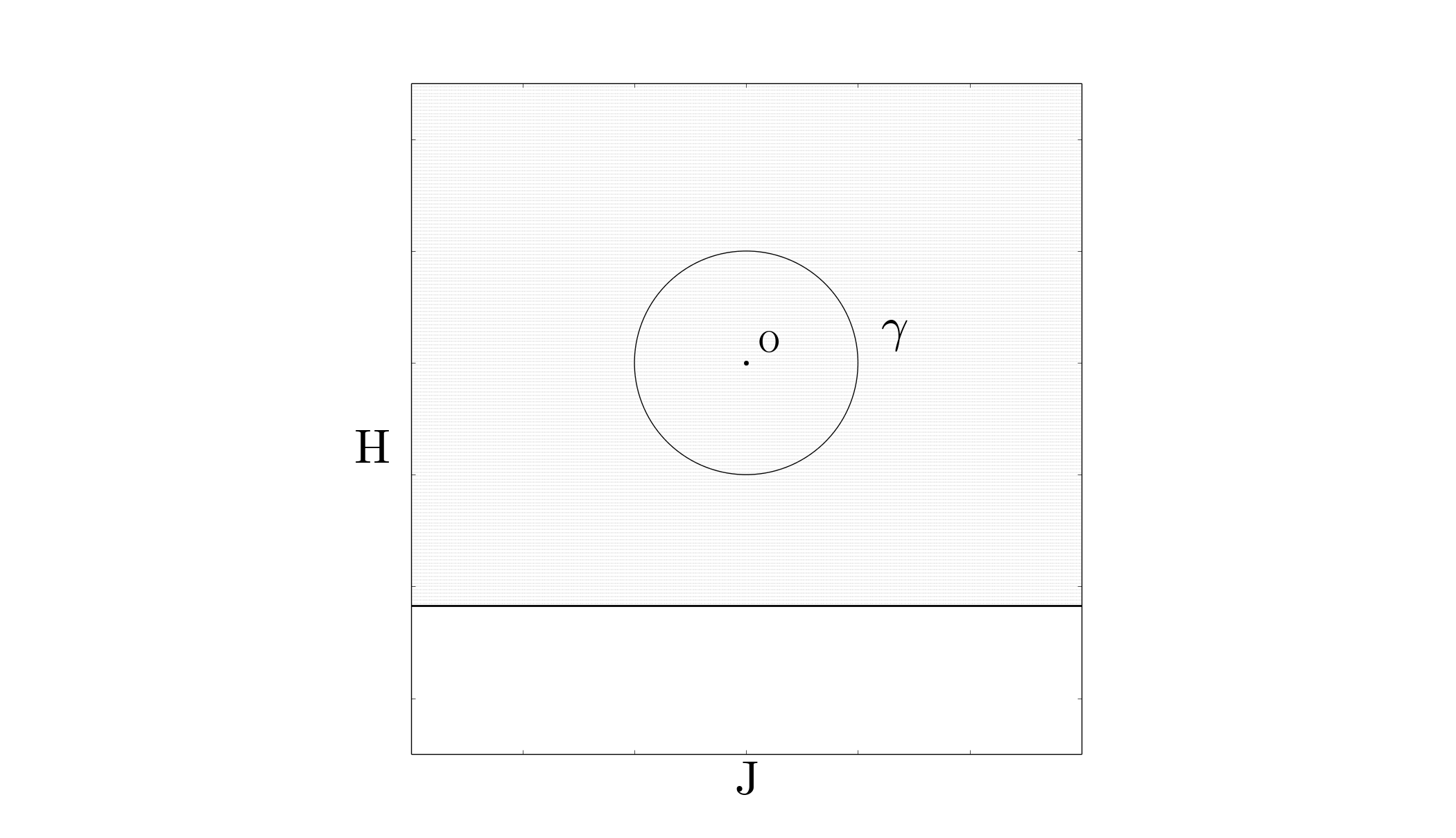}
\caption{Bifurcation diagram of a $1$:$(-1)$ system. The set of regular values is shown gray; 
the critical values are colored black; the isolated critical point $O = (0,0)$ lifts to the 
singly pinched torus $F^{-1}(O)$.}
  \label{1m1}
\end{figure}

\begin{remark}
If $mn<0$, then the fixed point $z = w = 0$ of the $\setS^1$ action is necessarily at the boundary of the corresponding bifurcation diagram.
Hence non-trivial monodromy (standard or fractional) can only be found when $mn>0$.
Because of Theorem~\ref{frmresonance}, non-trivial standard monodromy can manifest itself only when $m = n = 1$.
\end{remark}
\begin{example}
An example of such a $1$:$-1$ resonant system can be obtained by considering
the Hamiltonian
\begin{align*}
H &= p_1q_2+p_2q_1 + \varepsilon (q_1^2+p_1^2)(q_2^2+p_2^2).
\end{align*}
The bifurcation diagram of the integral map $F = (J,H)$ has the from shown in Fig.~\ref{1m1}.
From Theorem~\ref{frmresonance} we infer that the monodromy matrix along $\gamma$ has the form
$$\begin{pmatrix}
1 & 1 \\
0 & 1
\end{pmatrix}\in 
\textup{SL}(2,\setZ).$$
\end{example}

\begin{example}
An example of a $m$:$(-n)$ resonant system with non-trivial fractional monodromy is the specific $1$:$(-2)$ resonant system, which has been 
introduced in \cite{Nekhoroshev2006}. The system is obtained by considering the
 Hamiltonian 
\begin{equation*}
H = 2q_1p_1q_2 + (q_1^2 - p_1^2)p_2 + \varepsilon R(q,p)^2,
\end{equation*}
where $\varepsilon > 0$ and $R = R(q,p)$ is the $1{:}(2)$ oscillator.
The bifurcation diagram of the integral map $F = (J,H)$ has the form shown in Fig.~\ref{1stBD}. In this case the 
set of regular values is
simply connected and, thus, standard monodromy is trivial. Let the curve $\gamma$ be as in Fig.~\ref{1stBD}.
From Theorem~\ref{frmresonance} we infer that 
the parallel transport group $H^0_1$ is spanned by $2a$ and $b$, and that the
fractional monodromy matrix has the form
$$\begin{pmatrix}
1 & 1/2 \\
0 & 1
\end{pmatrix}\in \textup{SL}(2,\mathbb Q).$$
This system is discussed in greater detail in Subsection~\ref{sec/intrresonance}.
\end{example}

\subsection{A system on $S^2\times S^2$}

\label{s2crosss2}

Let $(x_1,x_2,x_3)$ and $(y_1,y_2,y_3)$ be coordinates in $\setR^3$. The relations 
$$\{x_i,x_j\} = \epsilon_{ijk}x_k, \ \{y_i,y_j\} = \epsilon_{ijk}y_k \mbox{ and } \{x_i,y_j\} = 0$$
define a Poisson structure on $\setR^3\times\setR^3$. The restriction of this Poisson structure
to $S^2\times S^2 = \{(x,y) \colon |x| = |y| = 1\}$ gives the canonical symplectic structure $\omega$.
 \begin{figure}[htbp]
  \includegraphics[width=1\linewidth]{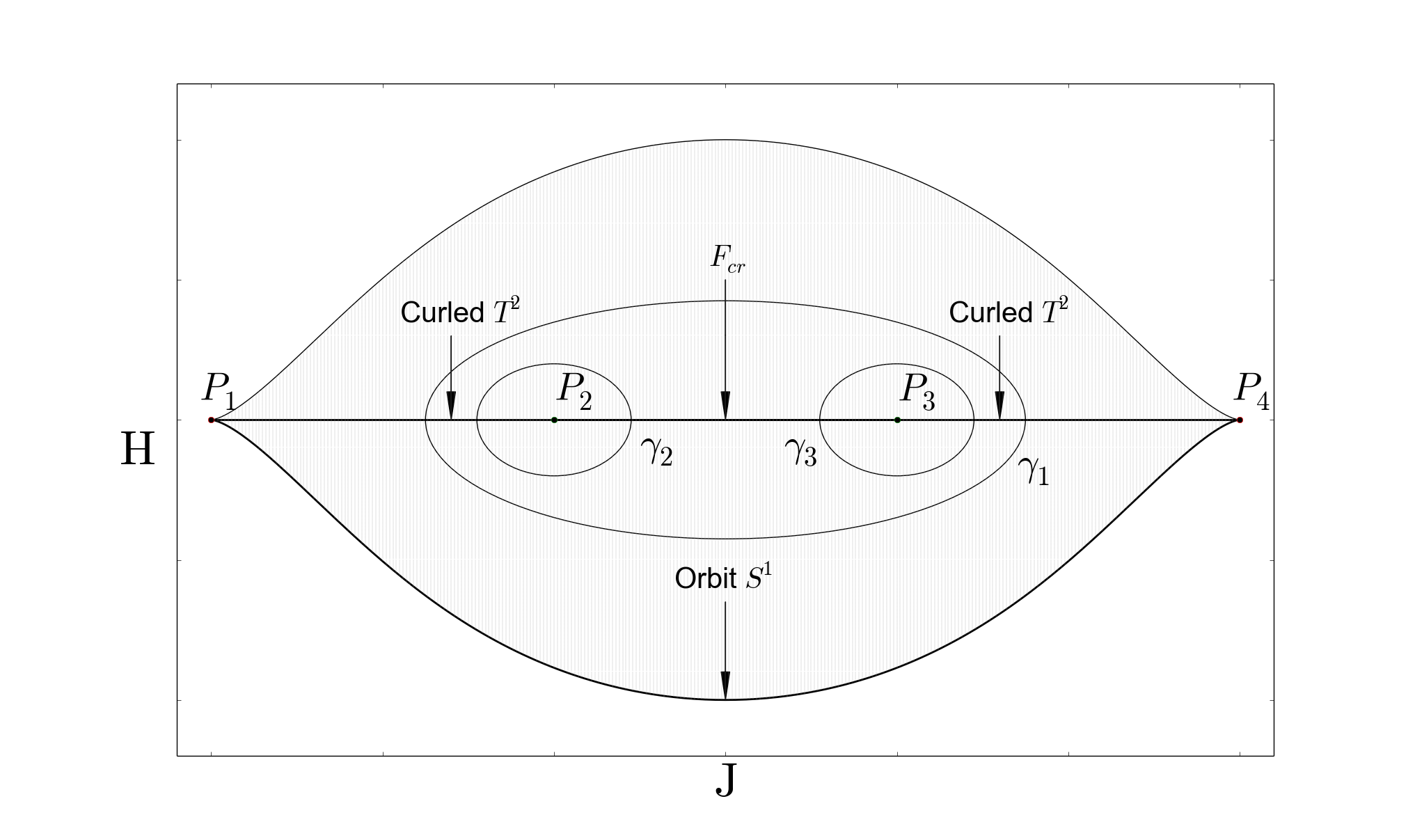}
  \caption{Bifurcation diagram of the integral map $F$. The set of regular
    values shown gray; the critical values are colored black. All regular fibers
    are $2$-tori. Curled $T^2$ contains one exceptional (`short')
    orbit of the $\setS^1$ action. Critical fibers $F_{cr}$ contain two such
    orbits. They can be obtained by gluing two curled tori along a regular orbit
    of the $\setS^1$ action.}
  \label{sphereBD}
\end{figure}

We consider an integrable Hamiltonian system on $(S^2\times S^2, \omega)$ defined by the integral map $F = (J,H) \colon S^2\times S^2 \to \setR^2,$ where
$$J = x_1 + 2y_1 \ \ \mbox{ and } \ \
H = \textup{Re}\{(x_2+ix_3)^2(y_2-iy_3)\}.$$
It is easily checked that the functions $J$ and $H$ commute, so $F$ is indeed an integral map. The bifurcation diagram is shown in
Fig.~\ref{sphereBD}.

Even without knowing the precise structure of critical fibers of $F$, we can compute fractional monodromy along curves
$\gamma_1, \gamma_2$ and $\gamma_3$, shown in Fig.~\ref{sphereBD}.
Specifically, assume that $\gamma_i(0) = \gamma_i(1)$ lifts to a regular torus.
\begin{theorem} For each $\gamma_i$, the parallel transport group is
spanned by $2a_i$ and $b_i$, where
$(a_i,b_i)$ forms a basis of $H_1(F^{-1}(\gamma_i(0))$ and $b_i$ is given by any orbit of the $\setS^1$ action. The fractional monodromy matrices
have the form
$$\begin{pmatrix}
1 & 1/2 \\
0 & 1
\end{pmatrix} \mbox{ for }  \ i = 2,3 \ \ \mbox{ and } \ \ \begin{pmatrix}
1 & 1 \\
0 & 1
\end{pmatrix}  \mbox{ for } \ i = 1.$$
\end{theorem}
\begin{proof}
 Consider the case $i = 2$. The other cases can be treated analogously. 
 The curve $\gamma_2$ intersects the critical line $H = 0$
 at two points $\xi_1$ and $\xi_2$. Let $\xi_1 < P_2 < \xi_2$ on $H = 0$. The critical fiber $F^{-1}(\xi_1)$, which is a curled torus, 
 contains  one exceptional
    orbit of the $\setS^1$ action with $\mathbb Z_2$ isotropy.
    The critical fiber $F^{-1}(\xi_1)$ contains two such
    orbits. Finally, observe that  the point 
    $$(1,0,0) \times (-1,0,0) \in S^2 \times S^2,$$ 
which    projects to $P_2$ under the map $F$,
    is fixed under the $\mathbb S^1$ action and has isotropy weights $m = 1, n = 2$. Since $F^{-1}(\gamma_2)$ is connected, it is left to apply
    Theorems~\ref{corexistence} and \ref{simplef}.
\end{proof}

\subsection{Revisiting the quadratic spherical pendulum}

The example of the system on $S^2 \times S^2$ discussed in the previous Subsection~\ref{s2crosss2} shows
that fractional monodromy matrix along a given curve $\gamma_1$ could be an integer matrix even if
standard monodromy along $\gamma_1$ is not defined. 
In this subsection we show that the same phenomenon can
appear when the isotropy groups are either trivial or $\setS^1,$ that is, when the 
$\setS^1$ action is free outside fixed points.

Consider a particle moving on the unit sphere 
\begin{align*}
\{ x = (x_1,x_2,x_3) \in \setR^3 \colon x_1^2 + x_2^2 + x_3^2 = 1 \}
\end{align*}
in a quadratic potential $V(x_3) = b x_3^2 + c x_3$.
The corresponding Hamiltonian system $(TS^2, \Omega|_{TS^2}, H)$, where $H(x,v) = \frac{1}{2}\langle v, v \rangle + V(x)$ is the total energy, is called \emph{quadratic spherical pendulum} \cite{Efstathiou2005}.
This system is completely integrable since the $x_3$ component $J$ of the angular momentum is conserved.
Moreover, $J$ generates a global Hamiltonian $\setS^1$ action on $TS^2$.
For a certain range of parameters $b$ and $c$ the bifurcation diagram of the integral map $F = (J,H)$ has the form shown in Fig.~\ref{qsp}.
 \begin{figure}[htbp]
  \includegraphics[width=1\linewidth]{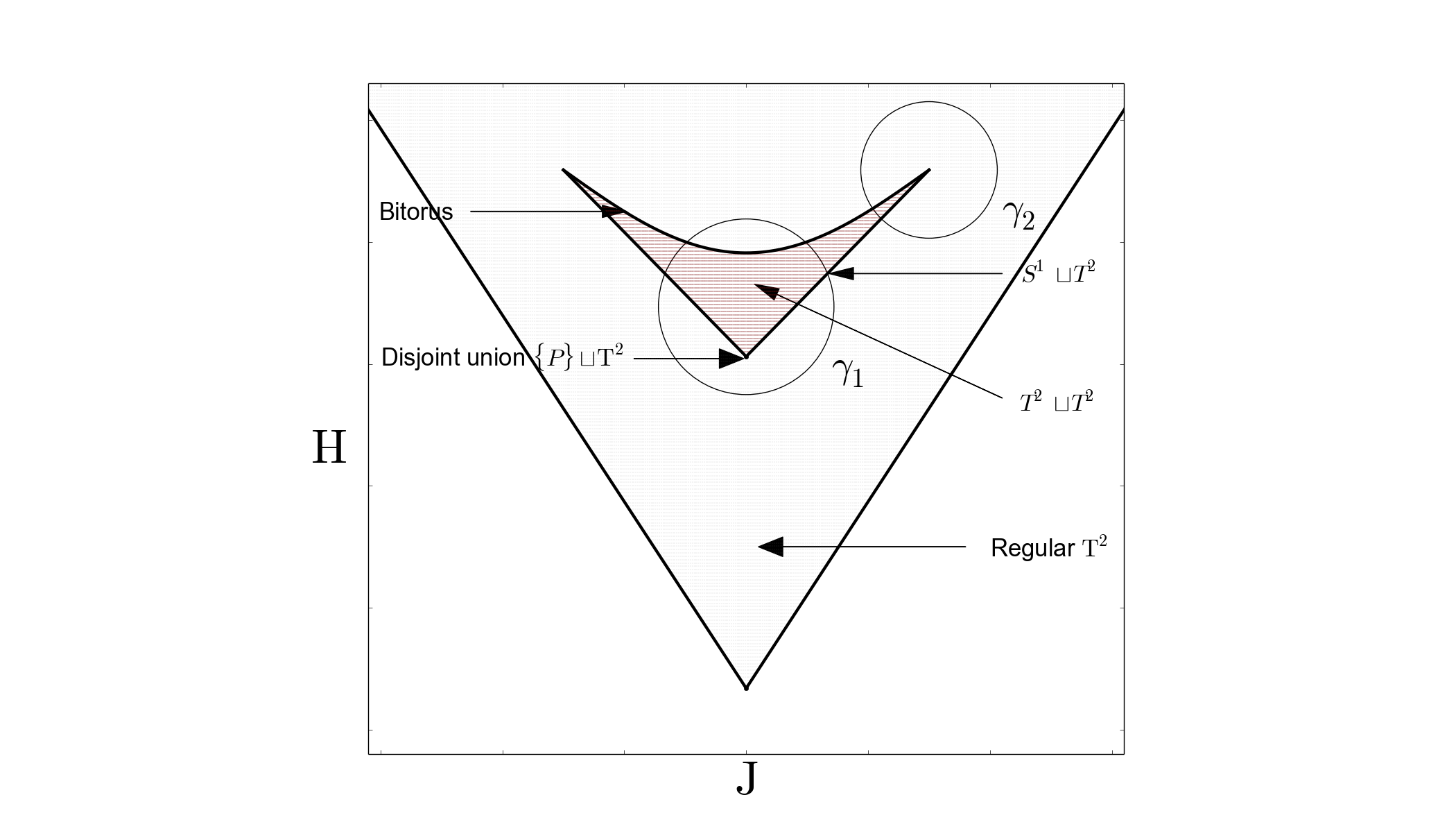}
  \caption{Bifurcation diagram of the integral map $F$. The set of regular
    values shown gray. The critical values are colored black. The points in the interior of the `island' are regular and lift
    to the disjoint union of $2$ tori.}
  \label{qsp}
\end{figure}

Let $\gamma_1$ and $\gamma_2$ be as in Fig.~\ref{qsp}. Assume that the starting point $\gamma_i(0) = \gamma_i(1)$ lifts to a regular torus.
\begin{theorem} \label{monqsp}
 For each $\gamma_i$, the parallel transport group coincides with the whole homology group $H_1(F^{-1}(\gamma_i(0))$. The fractional monodromy matrices have the form
 $$\begin{pmatrix}
1 & 1 \\
0 & 1
\end{pmatrix} \mbox{ for }  \ i = 1 \ \ \mbox{ and } \ \ \begin{pmatrix}
1 & 0 \\
0 & 1
\end{pmatrix}  \mbox{ for } \ i = 2.$$
\end{theorem}
\begin{proof}
 Consider the case $i = 1$. The other case can be treated similarly. 
 The $\mathbb S^1$ action  is free on the connected manifold $F^{-1}(\gamma_1)$. The  Euler number of this manifold equals $1$. Indeed,
 the elliptic-elliptic
 point 
 $$P = (0,0,1)\times(0,0,0) \in TS^2 \subset T\mathbb R^3,$$ 
which projects to the point $F(P) \in \textup{int}(\gamma_1)$, is fixed under the $\mathbb S^1$ action and has isotropy weights $m =1, n= 1$.
 It is left to apply Theorems~\ref{corexistence} and \ref{simplef}.
\end{proof}

\begin{remark}
 From Theorem~\ref{monqsp} it follows that all homology cycles can be parallel transported along $\gamma_i, \ i = 1,2.$ Even though this situation is 
 very similar
 to the case of standard monodromy, the monodromy along $\gamma_i$ is fractional. We note that such examples have not been considered until now.
\end{remark}

\section{Proof of Theorem~\ref{mth}} \label{seifert}

In the present section we use the notation introduced in Subsection~\ref{mresult}. The result, Theorem~\ref{mth}, will follow from 
Lemmas~\ref{lemma/existence}, \ref{lemma/uniqueness}, \ref{lemma/eulernumber}, and \ref{lemma/span} that are given below.
 \begin{figure}[htbp]
 \hspace{0cm}
  \includegraphics[width=0.7\linewidth]{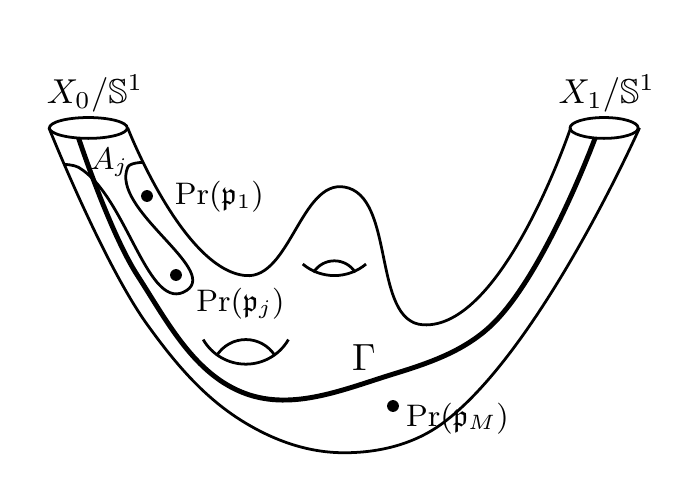}
  \caption{The base manifold $X / \setS^1$. }
  \label{base}
\end{figure} 
  
  \begin{lemma} \label{lemma/existence}
There exists $k \in \setZ$ such that
$(Na_0, Na_1 + kb_1)$ and $(b_0, b_1)$ belong to  $\partial_{*}(H_2(X, \partial X)).$
\end{lemma}
\begin{proof}
Let $\mathbb Z_N$ be the order $N$ subgroup  of $\mathbb S^1.$ 
The quotient $X' = X/\setZ_N$, which is given by the induced action of the subgroup $\mathbb Z_N$, is the total space of the
 principal circle bundle $$\Pr{'} \colon X' \to X / \mathbb S^1.$$ 
We note that this bundle is, moreover, trivial. Indeed, the base $X/\mathbb S^1$ has a boundary and is, thus, homotopy equivalent to a graph.
  
  Let $b_i^r = b_i / \setZ_N, \ i = 0,1$. Then $(a_i,b_i^r)$ forms a basis of $H_1(X_i / \setZ_N)$.
There is a unique parallel transport of the cycles $a_0$ and $b^r_0$ along $X'$. 
Indeed, take a global section $s \colon X'/\setS^1 \to X'$ with $s(X_0/\setS^1) = a_0.$ Then 
$S = s(X'/\setS^1)$
is a relative $2$-cycle that gives the parallel transport of $a_0$. 
In order to transport the cycle $b_0^r$ take a smooth curve 
$\Gamma \subset X/\setS^1$ connecting $X_0/\setS^1$ with $X_1/\setS^1$ and define the relative $2$-cycle by
$(\Pr{'})^{-1}(\Gamma).$

\begin{remark}
In what follows we assume that $\Gamma$ is a simple curve that does not contain the singular points $\Pr (\mathfrak p_1), \ldots, \Pr (\mathfrak p_M),$ 
where $\Pr \colon X \to X / \setS^1$ is the canonical projection; see Fig.~\ref{base}.
\end{remark}

From above it follows that the parallel transport in the reduced space has the form $
a_0 \mapsto a_1 + k b_1^r$ and $
b_0^r \mapsto b_1^r$ for some
 $k \in \mathbb Z$.
 The parallel transport of the cycles $(a_0,b_0^r)$ in the reduced space lifts to the parallel transport of the cycles $(Na_0,b_0)$ along $X$ 
 in the original space. 
Indeed, let $\pi \colon X \to X'$ be the quotient map, given by the action of $\mathbb Z_N$. The preimage 
$$\pi^{-1}((\Pr{'})^{-1}(\Gamma)) = \Pr{^{-1}}(\Gamma)$$ 
transports $b$ since $\Gamma$ does not contain the singular points $\Pr(\mathfrak p_j)$. In order to transport $Na$ take 
$\pi^{-1}(S).$ 
Since $\pi \colon \pi^{-1}(S) \to S$ is a branched $N$-covering, see Fig.~\ref{zet2}, the preimage $\pi^{-1}(S)$ is a relative $2$-cycle that transports $Na$.
The result follows.
\end{proof}
\begin{figure}[htbp]
\hspace{0cm}
  \includegraphics[width=0.36\linewidth]{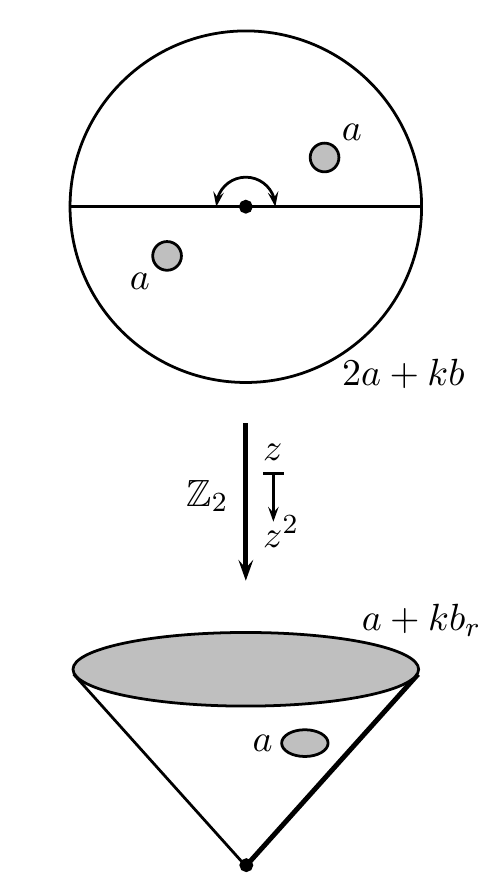}
  \caption{An example of the covering map $\pi \colon \pi^{-1}(S) \to S$. Here the Seifert manifold $X$ contains only one exceptional orbit with 
  $\setZ_2$ isotropy ($N = 2$); the base $X / \setS^1 \cong S$ is a `cone with a hole'. }
  \label{zet2}
\end{figure}

This following lemma shows that the parallel transport along $X$
is unique.

\begin{lemma} \label{lemma/uniqueness}
Suppose that
$(0,c) \in \partial_*(H_2(X,\partial X))$ for some $c \in H_1(X_1)$. Then we have $c = 0$. 
 \end{lemma}
\begin{proof}
 This statement was essentially proved in \cite{Efstathiou2013} (see \S 7.1 therein). 
 For the sake of completeness we provide a proof below.

Since $X$ is an orientable $3$-manifold, the rank of the image $\partial_*(H_2(X, \partial X))$ is half of the rank of $H_1(\partial X) \simeq \setZ^2 \oplus \setZ^2$. Hence
\begin{equation*}
\textup{rk} \ \partial_*(H_2(X, \partial X)) = 2.
\end{equation*}
As a subgroup of a free abelian group $H_1(\partial X)$, the image $\partial_*(H_2(X, \partial X))$ is a free abelian group and thus is isomorphic to $\setZ \oplus \setZ$. 

From Lemma~\ref{lemma/existence} we get that 
$\alpha = (Na_0, Na_1 + k b_1)$ and $\beta = (b_0,b_1)$ belong to 
$\partial_*(H_2(X, \partial X))$.

Suppose that parallel transport along $X$ is not unique. Then there exists an element $\eta = (0,c) \in \partial_*(H_2(X, \partial X))$ with $c \ne 0.$ 
Since  $\alpha$ and $\beta$ are linearly independent over $\setZ$, we get
$l_1 \alpha + l_2 \beta = l_3 \eta$, where $l_j$ are integers and $l_3 \ne 0$. But $l_1 N a_0 + l_2 b_0 = 0$, so $l_1 = l_2 = 0$ and we get a contradiction.
\end{proof}

The set $H^0_1$ of cycles $\alpha \in H_1(X_0)$ that can be parallel transported along $X$ forms a subgroup
of $H_1(X_0)$. Since $Na_0$ and $b_0$ can be parallel transported along $X$, the group
  $H^0_1$ is spanned by  $La_0$ and $b_0$ for some $L \in \setN$, which divides $N$.
  Our goal is to prove that $L = N$. The proof of this equality is based on the important Lemma~\ref{lemma/eulernumber} below.
  
    Let $E$ be a closed Seifert manifold which is obtained from $X$ by identifying the boundary tori $X_i$ via an orbit preserving diffeomorphism
  that sends $a_0$ to $a_1$ and $b_0$ to $b_1$.

  \begin{lemma} \label{lemma/eulernumber}
The Euler number $e(E)$ of the Seifert manifold $E$ satisfies $e(E) = \frac{k}{N}$.
\end{lemma}

\begin{proof}
Consider the action of the quotient circle $\setS^1 / \mathbb Z^N$ on the quotient space $E' = E/\mathbb Z^N$. Since $E'$ is a manifold and the action is free, we have a
 principal bundle $(E', B = E/ \mathbb S^1, \Pr{'})$.
 Let  
$$U_1 \cong [0,\varepsilon] \times S^1$$ 
be a cylindrical neighborhood of $X_0/\setS^1$ in $X/\setS^1$ with $\{0\}\times S^1 \cong X_0 / \setS^1.$
Define 
$$U_2 = \overline{B \setminus U_1}.$$ 
We already know that if $X' = X / \setZ_N$ then $(X', X/ \mathbb S^1, \Pr{'})$ is a trivial circle bundle.
Observe that $E'$ is obtained from $X'$
 by identifying the boundary tori $X'_0$ and $X'_1$ via a diffeomorphism induced by the `monodromy' matrix  
 $\begin{pmatrix}
            1 & k \\
            0 & 1
           \end{pmatrix}$. 
Hence there exist cross sections $s_1 \colon U_1 \to E'$ and $s_2 \colon U_2 \to E'$ such that $s_2 = s_1$ on the boundary circle
 $\{\varepsilon\}\times S^1$  and
$s_1 = e^{i k \varphi} s_2$ on  $\{0\}\times S^1$, parametrized by an angle $\varphi$.

Let $f \colon [0,2\pi] \to [0, 1]$ be a smooth function such that $f|_{[0,\delta]} = 1$ and $f|_{[2 \delta, 2 \pi]} = 0$. Define a continuous function $h \colon [0,\varepsilon] \times S^1 \to [0,2\delta]$ by the following formula
\begin{equation*}
h(\phi,\varphi)  = \dfrac{\varepsilon - \phi}{\varepsilon} \varphi f(\varphi) .
\end{equation*} 
Let $D^2 = (0, \varepsilon) \times (\delta, 2 \pi).$
Define new cross sections $s'_1 \colon U_1 \to E'$ and $s'_2 \colon B \setminus D^2 \to E'$  as follows
\begin{equation*}
s'_1 = s_1 \cdot e^{i k h } \  \ \  \mbox{ and } \ \ \ 
s'_2 = \begin{cases}
s_2&\text{on $U_2$,}\\
s'_1&\text{otherwise.}
\end{cases}
\end{equation*}

Observe that $s_1(0\times S^1) = s_2(0 \times S^1) + k b,$ where $b$ corresponds to the $\mathbb S^1$ action. If $\delta > 0$ is small enough, then $s_1(0\times S^1)$ is homological to 
$s'_1(0\times S^1).$ Hence 
\begin{equation*}
s'_1(0\times S^1) = s'_2(0 \times S^1) + k b.
\end{equation*}
But $s'_1(\partial D^2 + 0 \times S^1) = s'_2(\partial D^2 + 0 \times S^1)$. Therefore
\begin{equation*}
s'_2(\partial D^2) = s'_1(\partial D^2) + k b.
\end{equation*} 
Thus, $e(E') = k$ and  
 \begin{equation*}
 e(E) = \frac{1}{N}e(E') = \dfrac{k}{N}.
 \end{equation*}
\end{proof}

 \begin{lemma} \label{lemma/span}
  The parallel transport group $H^0_1$ is spanned by the cycles $Na_0$ and $b_0$.
 \end{lemma}
\begin{proof}
 We have already noted that 
  $H^0_1$ is spanned by $La_0$ and $b_0$ for some $L \in \mathbb N$, which divides $N$. In order to prove the equality $L = N$
  it is sufficient to prove that for every $j$ the number $L$ is a multiple of $n_j$ (the order of the exceptional orbit
  $\mathfrak p_j$). 
  
  The image of the exceptional fiber $\mathfrak p_j$ under the projection
$\Pr \colon E \to B = E/ \setS^1$ is a single point $\Pr (\mathfrak p_i)$ on the base manifold $B$. Cutting $E$ along the torus
$X_0 \cong X_1$ results in the manifold $X$. The quotient $X / \setS^1$ is obtained 
from $B$ by cutting along an embedded circle. Consider an annulus $A_j \subset  X / \setS^1$
that contains $X_0 / \setS^1$ and exactly one singular
point $\Pr (\mathfrak p_j)$; see Fig~\ref{base}.

Clearly, the preimage $E_j = \Pr^{-1}(A_j)$ is a Seifert manifold with only one exceptional fiber. From the 
definition of the parallel transport it follows that there exists a relative cycle $S  \subset E_j$ such that
one of the connected components of $S$ is $L a_0$. In other words, $La_0$ can be parallel transported along $E_j$.

Let us identify the boundary tori of $E_j$ via an orbit preserving diffeomorphism. Then the result of the parallel transport of $La_0$ along $E_j$ 
is $l_1 a_0 + l_2 b_0$. Since the parallel transport is unique, see Lemma~\ref{lemma/uniqueness}, we have
\begin{equation} \label{feq}
Nl_1 a_0 + N l_2 b_0 = NLa_0 = LNa_0 = LN a_0 + Lm_jb_0,
\end{equation}
 where $m_j \in \mathbb Z.$
Let $e_j$ denote the Euler number of the Seifert manifold $E_j$. From Lemma~\ref{lemma/eulernumber} 
it follows that  
$$m_j/n_j = e_j \pmod 1.$$ 
In particular, $m_j$ and $n_j$ are relatively prime.  Eq.~\eqref{feq} implies $N l_2 = L m_j$. Since $n_j$ divides $N$, it also divides $L$. 
\end{proof}

\section{Discussion} \label{discussion}

In \cite{Efstathiou2017} we have shown that if the circle action is 
free outside isolated fixed points then standard monodromy can be completely determined by the weights $1{:}(\pm1)$ of the circle action at those points.
This result allowed us to consider both focus-focus and elliptic-elliptic singular points of the integral map 
and provide a unified result for standard monodromy around such points.
Moreover, it showed that the circle action is more important for determining standard monodromy 
than the precise form of the integral map $F$.

In the present paper we generalized results from \cite{Efstathiou2017} to the
setting of Seifert fibrations. Specifically, we showed that the parallel transport along the total space of such a fibration is well defined 
and is completely determined by the Euler number and the orders of the exceptional orbits. Then, we applied the obtained results to fractional monodromy in singular Lagrangian fibrations (integrable Hamiltonian systems) 
that are
invariant under an effective (Hamiltonian) circle action with isolated fixed points. 

In the case of singular Lagrangian fibrations fixed points with weights $m{:}n$ different from
 $1{:}(\pm1)$ may appear.
The existence of such weights $m{:}n$
implies the existence of points with non-trivial isotropy group $\mathbb Z_m$ or $\mathbb Z_n$.
Such points are projected to one-parameter families of critical values of $F$. These families contain essential information about 
the geometry of the singular Lagrangian fibration. However,
for standard monodromy such critical families  are `invisible' in the sense that in 
standard monodromy we only consider the regular part of the fibration and the curves $\gamma$ along which standard monodromy is defined do not cross any critical values.
In the fractional case  the curves $\gamma$ are allowed to cross critical values of $F$. Our results show that also in this fractional case
the circle action is more important for fractional monodromy
than the precise form of the integral map $F$.

\section*{Acknowledgements}

We would like to thank Prof.\ Henk W.\ Broer for his valuable comments and suggestions on an early draft of this paper. We would also like to 
thank Prof.\ Gert 
Vegter and Prof.\ Holger Waalkens for useful discussions. We are grateful to the referee for his valuable comments which led to improvements of the
original version of this paper.   K. E. was partially
supported by the Jiangsu University Natural Science Research Program (grant
13KJB110026) and by the National Natural Science Foundation of China (grant
61502132).

\bibliographystyle{amsplain}
\bibliography{library}

\providecommand{\bysame}{\leavevmode\hbox to3em{\hrulefill}\thinspace}
\providecommand{\MR}{\relax\ifhmode\unskip\space\fi MR }
\providecommand{\MRhref}[2]{%
  \href{http://www.ams.org/mathscinet-getitem?mr=#1}{#2}
}
\providecommand{\href}[2]{#2}
\begin{thebibliography}{10}

\bibitem{Arnold1968}
V.~I. Arnol'd and A.~Avez, \emph{Ergodic problems of classical mechanics}, W.A.
  Benjamin, Inc., 1968.

\bibitem{Audin2004}
M.~Audin, \emph{Torus actions on symplectic manifolds}, Birkh{\"a}user, 2004.

\bibitem{Bates1991}
L.~M. Bates, \emph{Monodromy in the champagne bottle}, Journal of Applied
  Mathematics and Physics (ZAMP) \textbf{42} (1991), no.~6, 837--847.

\bibitem{Bates1993}
L.~M. Bates and M.~Zou, \emph{Degeneration of hamiltonian monodromy cycles},
  Nonlinearity \textbf{6} (1993), no.~2, 313--335.

\bibitem{Bochner1945}
S.~Bochner, \emph{Compact groups of differentiable transformations}, Ann. of
  Math. \textbf{46} (1945), no.~3, 372--381.

\bibitem{Bolsinov2012}
A.V. Bolsinov, Izosimov A.M., A.Y. Konyaev, and A.A. Oshemkov, \emph{Algebra
  and topology of integrable systems. {R}esearch problems (in {R}ussian)},
  Trudy Sem. Vektor. Tenzor. Anal. \textbf{28} (2012), 119--191.

\bibitem{Bolsinov2004}
A.V. Bolsinov and A.T. Fomenko, \emph{Integrable hamiltonian systems: Geometry,
  topology, classification}, CRC Press, 2004.

\bibitem{Broer2010}
H.W. Broer, K.~Efstathiou, and O.V. Lukina, \emph{A geometric fractional
  monodromy theorem}, Discrete and Continuous Dynamical Systems \textbf{3}
  (2010), no.~4, 517--532.

\bibitem{Cushman2015}
R.~H. Cushman and L.~M. Bates, \emph{Global aspects of classical integrable
  systems}, 2 ed., Birkh{\"a}user, 2015.

\bibitem{Cushman1985}
R.~H. Cushman and H.~Kn{\"o}rrer, \emph{The energy momentum mapping of the
  {L}agrange top}, Differential Geometric Methods in Mathematical Physics,
  Lecture Notes in Mathematics, vol. 1139, Springer, 1985, pp.~12--24.

\bibitem{Cushman2000}
R.H. Cushman and D.A. Sadovski{\'\i}, \emph{Monodromy in the hydrogen atom in
  crossed fields}, Physica D: Nonlinear Phenomena \textbf{142} (2000), no.~1-2,
  166--196.

\bibitem{Duistermaat1980}
J.~J. Duistermaat, \emph{On global action-angle coordinates}, Communications on
  Pure and Applied Mathematics \textbf{33} (1980), no.~6, 687--706.

\bibitem{Duistermaat1998}
\bysame, \emph{The monodromy in the hamiltonian hopf bifurcation}, Zeitschrift
  f\"{u}r Angewandte Mathematik und Physik (ZAMP) \textbf{49} (1998), no.~1,
  156.

\bibitem{Efstathiou2005}
K.~Efstathiou, \emph{Metamorphoses of {H}amiltonian systems with symmetries},
  Springer, Berlin Heidelberg New York, 2005.

\bibitem{Efstathiou2013}
K.~Efstathiou and H.~W. Broer, \emph{Uncovering fractional monodromy},
  Communications in Mathematical Physics \textbf{324} (2013), no.~2, 549--588.

\bibitem{Efstathiou2007}
K.~Efstathiou, R.H. Cushman, and D.A. Sadovskií, \emph{Fractional monodromy in
  the 1:−2 resonance}, Advances in Mathematics \textbf{209} (2007), no.~1,
  241 -- 273.

\bibitem{Efstathiou2017}
K.~Efstathiou and N.~Martynchuk, \emph{Monodromy of {H}amiltonian systems with
  complexity-1 torus actions}, Geometry and Physics \textbf{115} (2017),
  104--115.

\bibitem{Fomenko2010}
A.~T. Fomenko and S.~V. Matveev, \emph{Algorithmic and computer methods for
  three-manifolds}, 1st ed., Springer Netherlands, 1997.

\bibitem{Fomenko1990}
A.T. {Fomenko} and H.~{Zieschang}, \emph{{Topological invariant and a criterion
  for equivalence of integrable Hamiltonian systems with two degrees of
  freedom.}}, {Izv. Akad. Nauk SSSR, Ser. Mat.} \textbf{54} (1990), no.~3,
  546--575 (Russian).

\bibitem{Giacobbe2008}
A.~Giacobbe, \emph{Fractional monodromy: Parallel transport of homology
  cycles}, Diff. Geom. and Appl. \textbf{26} (2008), 140--150.

\bibitem{Hatcher2000}
A.~Hatcher, \emph{Notes on basic 3-manifold topology}, Available online, 2000.

\bibitem{Jankins1983}
M.~Jankins and W.D. Neumann, \emph{Lectures on seifert manifolds}, Brandeis
  lecture notes, Brandeis University, 1983.

\bibitem{Lerman1994}
L.~M. Lerman and Ya.~L. Umanski{\u \i}, \emph{Classification of
  four-dimensional integrable {H}amiltonian systems and {P}oisson actions of
  {$\mathbb{R}^2$} in extended neighborhoods of simple singular points. i},
  Russian Academy of Sciences. Sbornik Mathematics \textbf{77} (1994), no.~2,
  511--542.

\bibitem{Lukina2008}
O.~V. Lukina, F.~Takens, and H.~W. Broer, \emph{Global properties of integrable
  hamiltonian systems}, Regular and Chaotic Dynamics \textbf{13} (2008), no.~6,
  602--644.

\bibitem{Matveev1996}
V.~S. Matveev, \emph{Integrable {H}amiltonian system with two degrees of
  freedom. the topological structure of saturated neighbourhoods of points of
  focus-focus and saddle-saddle type}, Sbornik: Mathematics \textbf{187}
  (1996), no.~4, 495--524.

\bibitem{Nekhoroshev1972}
N.~N. Nekhoroshev, \emph{Action-angle variables, and their generalizations},
  Trans. Moscow Math. Soc. \textbf{26} (1972), 181--198.

\bibitem{Nekhoroshev2007}
\bysame, \emph{Fractional monodromy in the case of arbitrary resonances},
  Sbornik: Mathematics \textbf{198} (2007), no.~3, 383--424.

\bibitem{Nekhoroshev2006}
N.N. Nekhoroshev, D.A. Sadovski\'{i}, and B.I. Zhilinski\'{i}, \emph{Fractional
  {H}amiltonian monodromy}, Annales Henri Poincar\'{e} \textbf{7} (2006),
  1099--1211.

\bibitem{Sadovskii1999}
D.~A. Sadovski{\'\i} and B.~I. Zhilinski{\'\i}, \emph{Monodromy, diabolic
  points, and angular momentum coupling}, Physics Letters A \textbf{256}
  (1999), no.~4, 235--244.

\bibitem{Schmidt2010}
S.~Schmidt and H.~R. Dullin, \emph{Dynamics near the $p:-q$ resonance}, Physica
  D: Nonlinear Phenomena \textbf{239} (2010), no.~19, 1884--1891.

\bibitem{Sugny2008}
D.~Sugny, P.~Marde{\v s}i{\'c}, M.~Pelletier, A.~Jebrane, and H.~R. Jauslin,
  \emph{Fractional hamiltonian monodromy from a gauss--manin monodromy},
  Journal of Mathematical Physics \textbf{49} (2008), no.~4, 042701.

\bibitem{Tonkonog2013}
D.I. Tonkonog, \emph{A simple proof of the geometric fractional monodromy
  theorem}, Moscow University Mathematics Bulletin \textbf{68} (2013), no.~2,
  118--121.

\bibitem{Waalkens2004}
H.~Waalkens, H.~R. Dullin, and P.~H. Richter, \emph{The problem of two fixed
  centers: bifurcations, actions, monodromy}, Physica D: Nonlinear Phenomena
  \textbf{196} (2004), no.~3-4, 265--310.

\bibitem{Waalkens2003}
H.~Waalkens, A.~Junge, and H.~R. Dullin, \emph{Quantum monodromy in the
  two-centre problem}, Journal of Physics A: Mathematical and General
  \textbf{36} (2003), no.~20, L307.

\bibitem{Zung1997}
N.~T. Zung, \emph{A note on focus-focus singularities}, Differential Geometry
  and its Applications \textbf{7} (1997), no.~2, 123--130.

\end{thebibliography}

\end{document}